\documentclass[11pt,a4paper]{article}

\usepackage{amsmath,latexsym,amssymb,amsthm}
\usepackage[dvipsnames]{xcolor}
\usepackage{multirow,hhline,booktabs}
\usepackage{bm,indentfirst,comment,xr-hyper}
\usepackage{natbib}
\usepackage[colorlinks,citecolor=blue,urlcolor=blue]{hyperref}
\makeatletter
\def\section{\@startsection {section}{1}{\z@}{-3.5ex plus -1ex minus -.2ex}{2.3 ex plus .2ex}{\large\bf}}
\def\subsection{\@startsection {subsection}{1}{\z@}{-3.5ex plus -1ex minus -.2ex}{2.3 ex plus .2ex}{\large\it}}
\makeatother

\theoremstyle{definition}
\newtheorem{theorem}{Theorem}

\newtheorem{assumption}{Assumption}
\newtheorem{remark}{Remark}
\newcommand{\argmin}{\mathop{\rm argmin}}

\def\T{{\rm T}}
\def\E{{\rm E}}

\def\oP{{\rm o}_{\rm P}}
\def\OP{{\rm O}_{\rm P}}
\def\vech{{\rm vech}}

\renewcommand{\baselinestretch}{1.8}\selectfont

\allowdisplaybreaks

\begin{document}

\title{\vspace{-15mm} \bf \Large Covariate balancing estimation and model selection for difference-in-differences approach}
\author{Takamichi Baba\\
\small Department of Statistical Science, The Graduate University for Advanced Studies\\ \small Biostatistics Center, Shionogi \& Co. Ltd. \bigskip \\
Yoshiyuki Ninomiya\\
\small Department of Statistical Science, The Graduate University for Advanced Studies\\ \small Department of Fundamental Statistical Mathematics, The Institute of Statistical Mathematics}
\date{}

\maketitle

\begin{abstract}
Remarkable progress has been made in difference-in-differences (\MakeUppercase{did}) approaches to causal inference that estimate the average effect of a treatment on the treated (\MakeUppercase{att}). Of these, the semiparametric \MakeUppercase{did} (\MakeUppercase{sdid}) approach incorporates a propensity score analysis into the \MakeUppercase{did} setup. Supposing that the \MakeUppercase{att} is a function of covariates, we estimate it by weighting the inverse of the propensity score. In this study, as one way to make the estimation robust to the propensity score modeling, we incorporate covariate balancing. Then, by attentively constructing the moment conditions used in the covariate balancing, we show that the proposed estimator is doubly robust. In addition to the estimation, we also address model selection. In practice, covariate selection is an essential task in statistical analysis, but even in the basic setting of the \MakeUppercase{sdid} approach, there are no reasonable information criteria. Here, we derive a model selection criterion as an asymptotically bias-corrected estimator of risk based on the loss function used in the \MakeUppercase{sdid} estimation. We show that a penalty term can be derived that is considerably different from almost twice the number of parameters that often appears in \MakeUppercase{aic}-type information criteria.

\medskip

\noindent\textbf{Keywords}: Average effect of treatment on the treated, Doubly robustness, Information criterion, Moment condition, Propensity score analysis,  Semiparametric difference-in-differences, Statistical asymptotic theory, Weighted mean squared risk
\end{abstract}


\section{Introduction}
\label{sec0}
The difference-in-differences (\MakeUppercase{did}) approach is one of the most common methods used in fields such as economics and epidemiology to assess the effect of a particular intervention or treatment. In the basic setting, there are two groups: a treatment group receiving a specific treatment and a control group receiving no treatment or other treatment. The pre- and post-treatment outcomes of each group are measured and the change in the outcome for the control group is subtracted from the change in the outcome for the treatment group to estimate the average treatment effect on the treated (\MakeUppercase{att}). Under the parallel trend assumption, an estimation in the \MakeUppercase{did} approach is valid. Recently, \cite{abadie2005} dealt with inference under a more natural assumption that involves using information from the covariates, specifically a parallel trend assumption conditioned on the covariates. They further proposed a semiparametric \MakeUppercase{did} (\MakeUppercase{sdid}) approach that does not necessarily require modeling of the covariates to the outcome and that estimates the \MakeUppercase{att} by weighting the inverse of the propensity score.

With or without conditionality, the parallel trend assumption is key to the \MakeUppercase{did} approach. This assumption is one where the mean changes over time in the potential outcomes of the treatment group and the control group would have been similar if there had been no treatment. On the other hand, especially in epidemiology, the ignorable treatment assignment assumption has usually been invoked in causal inference (e.g. \citealt{robins94}).
This is the assumption that the potential outcome variable and the treatment assignment variable, which represents the received treatment, become independent conditioned on the covariates. Although neither assumption is verifiable, to evaluate the validity of the parallel trend assumption, visual inspection is useful when observations over many time points are available. Moreover, the parallel trend assumption naturally holds if it can be supposed that there is no change over time in the expectation of outcomes without treatment. Needless to say, the parallel trend assumption is considerably weaker than the ignorable treatment assignment assumption, and estimating the \MakeUppercase{att} under such a desirable assumption is the advantage that the \MakeUppercase{did} approach gains in exchange for collecting pre- and post-treatment outcomes.

The methodological development of the \MakeUppercase{did} approach has been growing because of the above-mentioned advantages and the increasing amount of available observational study data (see \citealt{roth23} for a comprehensive review of recent studies on the \MakeUppercase{did} approach). However, for example, \cite{abadie2005} assumes that the model for the propensity score is correctly specified. However, the estimator will have a bias if the model is misspecified.
As an extension of \cite{abadie2005}, \cite{sant2020} assumes not only a regression model for the treatment assignment variable used to obtain the propensity score, but also a regression model for the outcome variable, and proposes an estimator for the \MakeUppercase{att} derived from the two models. Their estimator is consistent if either the treatment assignment model or the outcome regression model is correctly specified, that is, the estimator is doubly robust.

In this article, as with \cite{abadie2005}, we suppose that the \MakeUppercase{att} is a function of covariates and address the semiparametric estimation problem for that function.
From the estimated function, the heterogeneity of the \MakeUppercase{att}, that is, the conditional \MakeUppercase{att} can be evaluated.
Two regression models as above are considered, and in order to give double robustness to the estimation, we incorporate covariate balancing (e.g. \citealt{egel08}, \citealt{ImaR14} and \citealt{FanILLNY22}), a method using propensity scores such that a certain moment condition is satisfied. This method actually estimates the treatment assignment model, but only supposes an outcome model and does not estimate it. Nevertheless, the \MakeUppercase{sdid} estimator derived by weighting the inverse of the propensity score, which is obtained by the covariate balancing, will be consistent if either of the regression models is correctly specified. One point to note is that, in order for the \MakeUppercase{sdid} approach to have this favorable asymptotic property, the moment condition has to be constructed in a different way from the usual  one, which is due to considering the conditional \MakeUppercase{att}. Also, the incorporation of covariate balancing leads to an estimation via weights that always balances the moments of the covariates across groups, which is also preferable from an interpretive point of view (\citealt{ImaR14}). We call the proposed method covariate balancing for \MakeUppercase{did} (\MakeUppercase{cbd}).

Obviously, model selection is also an essential task in using the \MakeUppercase{did} approach in a series of data-analysis tasks.
Specifically, by selecting the function of covariates for the \MakeUppercase{att} addressed in this paper, it becomes possible to evaluate the heterogeneity of the \MakeUppercase{att}. There have been various important studies conducted in this area.
For example, in a study estimating the effect of union membership on wages, \cite{houng16} treated the workers' demographics with a linear model for causal effects and explored factors affecting the causal effect on the basis of the values of the regression coefficients and its $p$-value. In other words, they addressed the issue of covariate selection. However, there is no satisfactory information criterion, a basic tool for model selection, even for the basic setting of \MakeUppercase{sdid}. To remedy this situation, we will derive an information criterion for a basic \MakeUppercase{sdid} approach and the \MakeUppercase{cbd} method that is an asymptotically unbiased estimator of the risk based on the loss function used in the estimation. Traditional information criteria such as the generalized information criterion (\MakeUppercase{gic}) (\citealt{KonK96}) cannot be applied because the loss function includes weights based on the propensity score, which are random variables. Since the \MakeUppercase{sdid} approach places importance on being based on the conditional parallel trend assumption, it does not rely on the ignorable treatment assignment assumption when deriving the information criterion, which causes some difficulties. As a result, we will derive a penalty term that is considerably different from twice the number of parameters, which often appears in \MakeUppercase{aic}-type information criteria including the \MakeUppercase{gic}.

In summary, the aim of this paper is to develop a new doubly robust estimation method with the advantage of covariate balancing in the \MakeUppercase{sdid} approach and to derive a model selection criterion for the \MakeUppercase{sdid} approach, for which no reasonable model selection criteria exist yet. Its organization is as follows. Section \ref{sec1} describes the model and the estimation target and introduces the \MakeUppercase{sdid} approach that was proposed in \cite{abadie2005}. In Section \ref{double}, we propose the \MakeUppercase{cbd} method and give a theorem showing its double robustness. Numerical experiments indicate that the estimation is robust against model misspecification. In Section \ref{sec3}, we derive a model selection criterion for when the propensity score is known as a bias-corrected estimator of the risk that is naturally defined from the loss function used in the \MakeUppercase{sdid} estimation. Then, we also prove a theorem that suggests that the penalty term of the model selection criterion is asymptotically valid. As a comparison, we derive a model selection criterion that naturally extends the idea of \MakeUppercase{qic}$_{\text{\MakeUppercase{w}}}$ in \cite{PlaBCWS13} to the \MakeUppercase{sdid} approach. Furthermore, through numerical experiments, we verify the approximation accuracy of the penalty term of the proposed criterion from the perspective of risk estimation and show its superior performance to the \MakeUppercase{qic}$_{\text{\MakeUppercase{w}}}$ from the viewpoint of risk minimization. Section \ref{sec5} summarizes the results for the \MakeUppercase{cbd} method in the same way as Section \ref{sec3}. In other words, a model selection criterion is derived from the naturally defined risk, a theorem supports the validity of its penalty term, and its performance is evaluated through numerical experiments. Section \ref{sec6} reports the results of the estimation and model selection for the real data in \cite{Lal86}. Finally, Section \ref{sec7} discusses future developments of the proposed approach.

\section{Preparation}
\label{sec1}
\subsection{Model and Assumptions}
\label{sec1_1}
Suppose that there are two groups, a treatment group and a control group, and that outcomes are collected before and after the possible treatment time point, respectively. We consider the counterfactual setting, where $y^{[1]}(t)$ denotes the outcome variable of the treatment group, $y^{[0]}(t)$ the outcome variable of the control group, $t=0$ before treatment and $t=1$ after treatment. In practice, only one of the outcomes of the treatment group or control group can be obtained. By preparing treatment assignment variables such that $d^{[1]}=1\ (d^{[0]}=0)$ if the sample is assigned to the treatment group and $d^{[0]}=1\ (d^{[1]}=0)$ if the sample is assigned to the control group, the observed outcome variable collected at each time point is expressed as
\begin{align*}
y(t) = d^{[1]}y^{[1]}(t)+d^{[0]}y^{[0]}(t).
\end{align*}
Denoting the change in the outcome variables at two time points in the treatment group as $\Delta^{[1]}=y^{[1]}(1)-y^{[1]}(0)$ and the change in the control group as $\Delta^{[0]}=y^{[0]}(1)-y^{[0]}(0)$, the observed change is expressed as
\begin{align*}
\Delta = d^{[1]}\Delta^{[1]}+d^{[0]}\Delta^{[0]}.
\end{align*}
Here, for the \MakeUppercase{att} conditioned on the covariate $x\ (\in\mathbb{R}^p)$,
\begin{align}
a(x) \equiv \E[y^{[1]}(1)-y^{[0]}(1)\mid x,d^{[1]}=1],
\label{ATTx}
\end{align}
we suppose $g(x;\theta)$ $(\theta \in \theta \subset \mathbb{R}^p)$ as a working model to approximate it, and for simplicity, set the working model to be a linear one $x^{\T}\theta$. Then, we estimate $\theta$ and give the \MakeUppercase{att} by taking the expectation with respect to $x$ in the treatment group.

Let us define the optimal value of $\theta$ as
\begin{align}
\theta^* \equiv \argmin_{\theta}\E[\{\E[y^{[1]}(1)-y^{[0]}(1)\mid x,d^{[1]}=1]-x^{\T}\theta\}^2\mid d^{[1]}=1].
\label{true}
\end{align}
\cite{abadie2005} shows that $\theta^*$ can be consistently estimated with the \MakeUppercase{sdid} approach by using propensity scores $e^{[1]}(x)\equiv {\rm pr}(d^{[1]}=1\mid x)$ and $e^{[0]}(x)\equiv{\rm pr}(d^{[0]}=1\mid x)$ under several assumptions, and therefore the \MakeUppercase{att} can be consistently estimated. The assumptions are as follows.

\begin{assumption}[Conditional parallel trend]\label{hypo1}
It holds that
\begin{align*}
\E[y^{[0]}(1)-y^{[0]}(0)\mid x,d^{[1]}=1]=\E[y^{[0]}(1)-y^{[0]}(0)\mid x,d^{[0]}=1], 
\end{align*}
which means that, under the expectation conditioned on $x$, the change in outcomes without treatment is the same regardless of whether it is the treatment group or the control group.
\end{assumption}

\begin{assumption}[No-anticipation]\label{hypo2}
It holds that
\begin{align*}
\E[y^{[1]}(0)-y^{[0]}(0)\mid x,d^{[1]}=1]=0, 
\end{align*}
which means that, on average, there is no treatment effect prior to its implementation.
\end{assumption}

\begin{assumption}[Strong positivity]\label{hypo3}
For some $\varepsilon > 0$, it holds that
\begin{align*}
\varepsilon<{\rm pr}(d^{[1]}=1\mid x)<1-\varepsilon.
\end{align*}
\end{assumption}

\noindent
\cite{abadie2005} showed that, by letting $d=(d^{[0]},d^{[1]})$ and using weights $\rho(d,x)\equiv d^{[1]}/e^{[1]}(x)-d^{[0]}/e^{[0]}(x)$,
\begin{align}
& \E[\rho(d,x)\Delta\mid x] = \E[\Delta^{[1]}\mid x,d^{[1]}=1]-\E[\Delta^{[0]}\mid x,d^{[0]}=1] \notag \\
& = \E[\Delta^{[1]}-\Delta^{[0]}\mid x,d^{[1]}=1]
 = \E[y^{[1]}(1)-y^{[0]}(1)\mid x,d^{[1]}=1]
 = a(x)
\label{abaprop}
\end{align}
holds. Assumption \ref{hypo1} is used in the second equality and Assumption \ref{hypo2} is used in the third equality. This equation means that, by multiplying the observed values of change with weights based on propensity scores, the conditional \MakeUppercase{att} can be estimated via the \MakeUppercase{did} approach. Then, from \eqref{true} and \eqref{abaprop}, the optimal solution of $\theta$ can be expressed as
\begin{align}
\theta^* = \argmin_{\theta}\E[e^{[1]}(x)\{\rho(d,x)\Delta-x^{\T}\theta\}^2].
\label{abat0}
\end{align}
Moreover, from this equation,
\begin{align}
& \E[e^{[1]}(x)x\{\rho(d,x)\Delta-x^{\T}\theta^*\}] \notag \\
& = \E[e^{[1]}(x)x\{\E[\rho(d,x)\Delta\mid x]-x^{\T}\theta^*\}]
 = \E[e^{[1]}(x)x\{a(x)-x^{\T}\theta^*\}]
 = 0
\label{cons}
\end{align}
is satisfied, and therefore, 
\begin{align}
\E[e^{[1]}(x)xa(x)] = \E[e^{[1]}(x)xx^{\T}]\theta^*
\label{theta}
\end{align}
is obtained.

\subsection{Semiparametric estimation}
\label{sec1_2}
Here, for an independently and identically distributed set of samples of size $n$, the $i$-th sample will be denoted by the subscript $i$. \cite{abadie2005} proposed
\begin{align}
\hat{\theta} = \bigg\{\sum_{i=1}^{n}e^{[1]}(x_i)x_{i}x_{i}^{\T}\bigg\}^{-1}\sum_{i=1}^{n}e^{[1]}(x_i)x_{i}\rho(d_i,x_i)\Delta_i
\label{theta_est}
\end{align}
as a natural estimator of $\theta^*$ by solving \eqref{abat0}. When the observed values of change are multiplied by weights based on propensity scores, a term corresponding to the difference-in-differences appears; thus, this method is called a semiparametric \MakeUppercase{did} (\MakeUppercase{sdid}) estimator. The following regularity conditions are assumed.
\begin{assumption}\label{hypo4}
$(\mathrm{i})$ $\theta^*$ is in the interior of the compact set $\theta \subset \mathbb{R}^r$.
$(\mathrm{ii})$ $\E[y(t)^2]<\infty$ for all $t$. 
$(\mathrm{iii})$ $\E[x]$ is bounded and $\E[e^{[1]}(x)xx^{\T}]$ is non-singular.
\end{assumption}
\noindent
From Assumption \ref{hypo4} and \eqref{theta}, we obtain
\begin{align*}
& \hat{\theta} = \bigg\{\frac{1}{n}\sum_{i=1}^{n}e^{[1]}(x_i)x_{i}x_{i}^{\T}\bigg\}^{-1}\frac{1}{n}\sum_{i=1}^{n}e^{[1]}(x_i)x_{i}\rho(d_i,x_i)\Delta_i \\
& \stackrel{p}{\to} \E[e^{[1]}(x)xx^{\T}]^{-1} \E[e^{[1]}(x)x\rho(d,x)\Delta] \\ 
& = \E[e^{[1]}(x)xx^{\T}]^{-1} \E[e^{[1]}(x)xa(x)]
 = \E[e^{[1]}(x)xx^{\T}]^{-1} \E[e^{[1]}(x)xx^{\T}]\theta^*
 = \theta^*,
\end{align*}
and therefore $\hat{\theta}$ is a consistent estimator of $\theta^*$. In practice, the propensity scores are often unknown. In such cases, the propensity scores can be modeled as $e^{[1]}(x;\alpha)={\rm logit}(x^{\T}\alpha)$; for example, the parameter $\alpha$ can be estimated as $\hat{\alpha}$ by maximum likelihood estimation (\MakeUppercase{mle}), and the propensity score in \eqref{theta_est} can be replaced by $e^{[1]}(x;\hat{\alpha})$ to obtain the \MakeUppercase{sdid} estimator. If the propensity score model is correctly specified, i.e., it can be written as
\begin{align*}
{\rm pr}(d^{[1]}=1\mid x)={\rm logit}(x^{\T}\alpha^*),
\end{align*}
and if $\hat{\alpha}$ converges to $\alpha^*$ under some further assumptions, then $\hat{\theta}$ is still a consistent estimator of $\theta^*$.

\section{Covariate balancing for difference-in-differences approach}
\label{double}
\subsection{Estimation method and its statistical properties}
\label{theory}
Propensity scores are often estimated using \MakeUppercase{mle}, but if the model for the propensity scores is misspecified, the estimator of the \MakeUppercase{att} will be biased. Here, we propose a robust method that can deal with model misspecification. While a doubly robust estimation known as augmented inverse-probability-weighted estimation, which uses estimators from an outcome regression model, has been addressed in \cite{sant2020} and \cite{NinPT20}, here we construct a doubly robust estimation by incorporating covariate balancing (e.g. \citealt{egel08}, \citealt{ImaR14} and \citealt{FanILLNY22}). Let ${\rm O}$ be a zero matrix and let us define
\begin{align*}
& H^{[1]}(d,x;\alpha)\equiv e^{[1]}(x;\alpha)\bigg\{\frac{d^{[1]}}{e^{[1]}(x;\alpha)}-1\bigg\}xx^{\T},
\\
& H^{[0]}(d,x;\alpha)\equiv e^{[1]}(x;\alpha)\bigg\{\frac{d^{[0]}}{e^{[0]}(x;\alpha)}-1\bigg\}xx^{\T},
\end{align*}
and consider the moment conditions, $\E[H^{[1]}(d,x;\alpha)]= {\rm O}$ and $\E[H^{[0]}(d,x;\alpha)] = {\rm O}$.
From these conditions is obtained $\E[d^{[1]}xx^{\T}]=\E[\{e^{[1]}(x;\alpha)/e^{[0]}(x;\alpha)\}d^{[0]}xx^{\T}]$, which means that the expectation of $xx^{\T}$ in the treatment group is adjusted to match the expectation of $xx^{\T}$ in the control group through weighting based on the propensity scores. In the original covariate balancing, the first-order moments of the covariates are balanced; that is, $xx^{\T}$ in the expectation is naturally set to $x$. It should be noted, however, that the converse does not necessarily hold. In this paper, the second-order moments of the covariates are balanced.

The semi-vectorization of the upper triangular part of the matrix is written as $\vech(\cdot)$. Since $xx^{\T}$ is a symmetric matrix, by letting
\begin{align*}
h(d,x;\alpha) \equiv (\vech\{H^{[1]}(d,x;\alpha)\}^{\T},\vech\{H^{[0]}(d,x;\alpha)\}^{\T})^{\T},
\end{align*}
the moment conditions can be rewritten as
\begin{align}
\E[h(d,x;\alpha)]=0.
\label{exmoment}
\end{align}
We will denote an $\alpha$ satisfying $\eqref{exmoment}$ as $\alpha^{\dagger}$. Since the propensity score model may be misspecified, $\alpha^{\dagger}$ is not necessarily the true value $\alpha^{*}$. In practice, the \MakeUppercase{sdid} approach uses a $\hat{\alpha}$ satisfying the empirical version of the moment conditions,
\begin{align}
\frac{1}{n}\sum_{i=1}^{n}h(d_i,x_i;\hat{\alpha})=0.
\label{moment}
\end{align}
If the moment conditions are just-identified, that is, the dimension of the moment conditions (the number of constraints) is the same as the dimension of the parameter $\alpha$, then $\alpha$ that satisfies $\eqref{moment}$ exists. On the other hand, if they are over-identified, that is, the dimension of the moment conditions is larger than the dimension of $\alpha$, then an $\hat{\alpha}$ that satisfies \eqref{moment} does not necessarily exist. In that case, an $\hat{\alpha}$ satisfying
\begin{align*}
\frac{1}{n}\sum_{i=1}^{n}h(d_i,x_i;\hat{\alpha})\approx0
\end{align*}
is used. Specifically, $\hat{\alpha}$ is obtained from the generalized method of moments (\citealt{hansen1982}) or the generalized empirical likelihood method (\citealt{qin94}).

Here, we will use the generalized method of moments and propose $\hat{\alpha}^{\rm \MakeUppercase{cb}}\equiv\argmin_{\alpha}\{h_n(\alpha)^{\T}\allowbreak W_{n}\allowbreak h_n(\alpha)\}$ as the estimator of $\alpha$ by using the covariate balancing (\MakeUppercase{cb}) method, where $h_n(\alpha)\equiv n^{-1}\sum_{i=1}^{n}h(d_i,x_i;\alpha)$ and $W_{n}$ is an arbitrary semi-positive definite matrix. Upon defining $G_n(\alpha) \equiv \partial h_n(\alpha)/\partial \alpha^{\T}$, one sees that the estimator $\hat{\alpha}^{\rm \MakeUppercase{cb}}$ satisfies
\begin{align}
G_n(\hat{\alpha}^{\rm \MakeUppercase{cb}})^{\T}W_{n}h_n(\hat{\alpha}^{\rm \MakeUppercase{cb}})=0.
\label{gmmeq}
\end{align}
Its simplest candidate is the identity matrix, but the optimal matrix, $\{n^{-1}\sum_{i=1}^{n}h(d_i,x_i;\alpha)\allowbreak h(d_i,\allowbreak x_i;\alpha)^{\T}\}^{-1}$, provides the most efficient estimator of $\alpha$. However, this optimal matrix is not necessarily practical since it may be too close to being a degenerate matrix for some data, causing considerable instability in the estimation of $\alpha$ and consequently $\theta$.

\begin{remark}
The idea of using constraints that approximately satisfy the empirical version of the moment conditions to estimate the parameters of the propensity scores has also been examined. In \cite{NinPI20} a doubly robust estimation method is proposed for dealing with the case of high-dimensional covariates. This method combines a semiparametric estimation of the \MakeUppercase{att} with sparse estimation. It is assumed that the outcome can be expressed by a linear model of covariates, and to speed up the convergence of the Horvitz-Thompson estimator, the weak covariate balancing property
is included as a constraint in estimating the parameter of the propensity scores. 
\end{remark}

Here, we state the regularity conditions of the generalized method of moments (\citealt{newey94}) that are necessary to derive the properties of the proposed estimation method.

\begin{assumption}[Consistency]\label{hypo5}
$(\mathrm{i})$ $W_{n}$ converges in probability to $W$, which is a semi-positive definite matrix.  Also, $\alpha=\alpha^{\dagger}$ if and only if $W\E[h(d,x;\alpha)]=0$. $(\mathrm{ii})$ $\alpha^{\dagger}$ is in the interior of the compact set $A \subset \mathbb{R}^p$. $(\mathrm{iii})$ $h(d,x;\alpha)$ is continuous; that is, $e^{[1]}(x;\alpha)$ is continuous at each $\alpha\in A$ with probability one. $(\mathrm{iv})$ $\E[\sup_{\alpha\in A}\|h(d,x;\alpha)\|]<\infty$. 
\end{assumption}

\begin{assumption}[Asymptotic normality]\label{hypo6}
$(\mathrm{i})$ $h(d,x;\alpha)$ is continuously differentiable in a neighborhood $\Gamma$ of $\alpha^{\dagger}$ with probability approaching one. $(\mathrm{ii})$ $\E[h(d,x;\alpha^{\dagger})]=0$ and $\E[\|h(d,x;\alpha^{\dagger})\|^2]\allowbreak<\infty$.  $(\mathrm{iii})$ $\E[\sup_{\alpha\in\Gamma}\allowbreak\|\partial h(d,x;\alpha)/\partial\alpha^{\T}\|_{\rm F}]<\infty$.  $(\mathrm{iv})$ $G(\alpha^{\dagger})^{\T}W\allowbreak G(\alpha^{\dagger})$ is non-singular for $G(\alpha) \equiv \E[\partial h(d,x;\alpha)/\partial\alpha^{\T}]$.
\end{assumption}

\noindent
Under these regularity conditions, the proposed method can be shown to be doubly robust, as follows. Here, we name the estimator obtained by replacing the propensity scores with $e^{[1]}(x;\hat{\alpha}^{\rm \MakeUppercase{cb}})$ in \eqref{theta_est} the covariate balancing for \MakeUppercase{did} (\MakeUppercase{cbd}) estimator and denote it as $\hat{\theta}^{\rm \MakeUppercase{cbd}}$.

\begin{theorem}
If the propensity score model $e^{[1]}(x;\alpha)$ is correctly specified, or if the change in outcomes follows a model in covariates of the form
\begin{align} 
\E[\Delta^{[k]}\mid x,d^{[k]}=1]=x^{\T}\beta^{[k]*}+\kappa(x), 
\label{linear} 
\end{align} 
where $\kappa(x)$ is an unspecified function, then under Assumptions \ref{hypo1}-\ref{hypo6}, $\hat{\theta}^{\rm \MakeUppercase{cbd}}$ converges in probability to $\theta^*$. That is, the proposed estimator has double robustness.
\end{theorem}

\begin{proof}
From Assumption \ref{hypo5}, $\hat{\alpha}^{\rm \MakeUppercase{cb}}$ is a consistent estimator of $\alpha^{\dagger}$. By expanding $h_n(\hat{\alpha})$ around $\alpha^{\dagger}$ and applying Assumption \ref{hypo6} and \eqref{gmmeq}, it is shown to be
\begin{align}
\sqrt{\vphantom{x}}{n}(\hat{\alpha}^{\rm \MakeUppercase{cb}}-\alpha^{\dagger})
& = -\{G_n(\hat{\alpha}^{\rm \MakeUppercase{cb}})^{\T}W_{n}G_n(\hat{\alpha}^{\rm \MakeUppercase{cb}})\}^{-1}G_n(\hat{\alpha}^{\rm \MakeUppercase{cb}})^{\T}W_{n}\sqrt{\vphantom{x}}{n}h_n(\alpha^{\dagger})+\oP(1)
 \label{adiff} \\
& \stackrel{\rm d}{\to} {\rm N}(0,\{G(\alpha^{\dagger})^{\T}WG(\alpha^{\dagger})\}^{-1}G(\alpha^{\dagger})^{\T}W\Omega(\alpha^{\dagger})WG(\alpha^{\dagger})\{G(\alpha^{\dagger})^{\T}WG(\alpha^{\dagger})\}^{-1}), \notag
\end{align}
where $\Omega(\alpha)\equiv \E[h(d,x;\alpha)h(d,x;\alpha)^{\T}]$. From the consistency of $\hat{\alpha}^{\rm \MakeUppercase{cb}}$, we also obtain
\begin{align}
&\hat{\theta}^{\rm \MakeUppercase{cbd}}-\theta^* \notag \\
&= \bigg\{\frac{1}{n}\sum_{i=1}^ne^{[1]}(x_i;\hat{\alpha}^{\rm \MakeUppercase{cb}})x_ix_i^{\T}\bigg\}^{-1} \frac{1}{n}\sum_{i=1}^ne^{[1]}(x_i;\hat{\alpha}^{\rm \MakeUppercase{cb}})x_i\{\rho(d_i,x_i;\hat{\alpha}^{\rm \MakeUppercase{cb}})\Delta_i-x_i^{\T}\theta^*\} \label{expand} \\
&= \bigg\{\frac{1}{n}\sum_{i=1}^ne^{[1]}(x_i;\alpha^{\dagger})x_ix_i^{\T}\bigg\}^{-1} \frac{1}{n}\sum_{i=1}^ne^{[1]}(x_i;\alpha^{\dagger})x_i\{\rho(d_i,x_i;\alpha^{\dagger})\Delta_i-x_i^{\T}\theta^*\} + \oP(1). \label{expand2}
\end{align}
If the propensity score model is correctly specified and $\alpha^{\dagger}=\alpha^*$ is satisfied, that is, if  $e^{[1]}(x)=e^{[1]}(x;\alpha^{\dagger})$ and $e^{[0]}(x)=e^{[0]}(x;\alpha^{\dagger})$ hold, then from \eqref{cons}, $\hat{\theta}^{\rm \MakeUppercase{cbd}}$ is a consistent estimator of $\theta^*$. On the other hand, even if the propensity score model is misspecified, when the change in outcomes follows a linear model of covariates, that is, when \eqref{linear} holds, we obtain
\begin{align*}
a(x) = \E[\Delta^{[1]}\mid x,d^{[1]}=1]-\E[\Delta^{[0]}\mid x,d^{[0]}=1] = x^{\T}(\beta^{[1]*}-\beta^{[0]*})
\end{align*}
from \eqref{abaprop}. Then, it follows from \eqref{theta} that 
\begin{align*}
\theta^* = \E[e^{[1]}(x)xx^{\T}]^{-1} \E[e^{[1]}(x)xa(x)] = \beta^{[1]*}-\beta^{[0]*}.
\end{align*}
Therefore, we obtain
\begin{align*}
& \E[e^{[1]}(x;\alpha^{\dagger})x\{\rho(d,x;\alpha^{\dagger})\Delta-x^{\T}\theta^*\}] \\
& = \E\bigg[e^{[1]}(x_i;\alpha^{\dagger})x_i\bigg\{\frac{d^{[1]}_{i}}{e^{[1]}(x_i;\alpha^{\dagger})}\E[\Delta_{i}^{[1]}\mid x_i,d^{[1]}_{i}=1] \\
& \phantom{= \E\bigg[e^{[1]}(x_i;\alpha^{\dagger})x_i\bigg\{}
-\frac{d^{[0]}_{i}}{e^{[0]}(x_i;\alpha^{\dagger})}\E[\Delta_{i}^{[0]}\mid x_i,d^{[0]}_{i}=1]-x_i^{\T}\theta^*\bigg\}\bigg] \\
& = \E\bigg[e^{[1]}(x_i;\alpha^{\dagger})x_i\bigg[\bigg\{\frac{d^{[1]}_{i}}{e^{[1]}(x_i;\alpha^{\dagger})}-1\bigg\}x_i^{\T}\beta^{[1]*}-\bigg\{\frac{d^{[0]}_{i}}{e^{[0]}(x_i;\alpha^{\dagger})}-1\bigg\}x_i^{\T}\beta^{[0]*}\bigg]\bigg] \\
& = \E[h_1(d_i,x_i;\alpha^{\dagger})\beta^{[1]*}-h_0(d_i,x_i;\alpha^{\dagger})\beta^{[0]*}]
 = 0.
\end{align*}
Here, the last equality holds from the moment conditions in \eqref{exmoment}. From the above, when the propensity score model is correctly specified, or when the change in outcomes follows a linear model of covariates, it holds that
\begin{align*}
\sum_{i=1}^ne^{[1]}(x_i;\alpha^{\dagger})x_i\{\rho(d_i,x_i;\alpha^{\dagger})\Delta_i-x_i^{\T}\theta^*\}=\OP{(\sqrt{\vphantom{x}}{n})}
\end{align*}
in \eqref{expand2}, and therefore $\hat{\theta}^{\rm \MakeUppercase{cbd}}$ is consistent.
\end{proof}

\begin{remark}
Let us discuss the relationship between our proposed method and the method in \cite{sant2020}.
Unlike our setting, which estimates the conditional \MakeUppercase{att}, the latter directly estimates the \MakeUppercase{att} using a doubly robust approach.
In their approach as well, the propensity scores are estimated so that the first-order moments of the distribution for the covariates are balanced.
They also suppose a parametric model of the covariates for the change in the control group's outcomes and estimate it, using these estimators to estimate the \MakeUppercase{att}.
In contrast, we place importance on evaluating the heterogeneity of the \MakeUppercase{att}, including variable selection, and then set a function of the covariates for the \MakeUppercase{att}.
We suppose a parametric model of the covariates for the change in outcomes as well, although not estimating it.
Then, if the propensity scores are estimated so that not the first-order but the second-order moments of the distribution for the covariates are balanced, it is easily, yet somewhat unexpectedly, shown that the conditional \MakeUppercase{att} estimator is doubly robust.
\end{remark}

\begin{remark}
After the initial posting of our preprint, a related manuscript by \cite{LiM25} became publicly available.
That work studies the semiparametric local efficiency of estimating the unconditional \MakeUppercase{att} via covariate balancing in the \MakeUppercase{sdid} approach, building on \cite{FanILLNY22}. 
The emphasis there is primarily theoretical and provides valuable insights into efficiency considerations.
In contrast, Theorem 1 above highlights a methodological contribution: it shows that doubly robust estimation of the conditional \MakeUppercase{att} necessitates balancing second-order moments of the covariate distribution.
A detailed investigation of efficiency properties is beyond the scope of the present paper.
While the efficiency in the case where \eqref{linear} holds can be derived following \cite{FanILLNY22}, efficiency analysis when \eqref{linear} does not hold would require more advanced semiparametric tools, such as those developed in \cite{GraP22}.
\end{remark}

\subsection{Numerical experiments}
\label{robust}
Let us examine the robustness of the proposed estimation against model misspecification through simulation studies. The setting for data generation follows
\begin{align*}
& x_1, x_2 \sim{\rm Uniform}(0, 2), \quad d^{[1]}\sim{\rm Bernoulli}({\rm logit}(-x_1+\alpha^* x_2)),
\\
& y^{[0]}(0)=y^{[1]}(0)\sim{\rm N}(0,1), \quad \epsilon^{[0]}, \epsilon^{[1]} \sim{\rm N}(0,1),
\\
& y^{[0]}(1)=y^{[0]}(0)+\epsilon^{[0]}, \quad y^{[1]}(1)=y^{[1]}(0)+\beta^* x_1+\epsilon^{[1]}.
\end{align*}
Throughout this paper, we will suppose that all random numbers are generated independently. The parameter $\beta^*$ represents the magnitude of the contribution from the covariate $x_1$ to the causal effect. The propensity score is estimated using only $x_1$; that is, the parameter $\alpha^*$ represents the magnitude of the contribution from $x_2$ to the true propensity scores, which is not included in the model for estimating the propensity scores. When $\alpha^* \neq 0$, $x_2$ should also be included in the model, and we are doing a model misspecification because it was not included. Anyway, since \eqref{linear} holds, the proposed estimation is theoretically guaranteed to be consistent. In this setting, we will compare the proposed method with a naive method estimating the propensity score by the \MakeUppercase{mle} and check whether the proposed method is robust to misspecification of the propensity score model. Each of the numerical experiments described below was repeated 3,000 times.

The true values of the \MakeUppercase{att}, the mean estimates, and their empirical confidence intervals are listed in Table \ref{tab4} for several $\alpha^*$, $\beta^*$ and sample sizes $n$. With increasing $\alpha^*$, which represents the magnitude of the misspecification of the propensity score model, the bias of the naive method becomes larger, and the true value is eventually no longer included in the confidence interval. Regarding the proposed method, the approximation accuracy is considerably high when the identity matrix is used as the weighting matrix in the generalized method of moments. On the other hand, when the optimal matrix is used as the weighting matrix, although the accuracy is often comparable to that of using the identity matrix, in some cases, it is similar to that of the naive method. Since there are settings in simulation studies where the optimal matrix is close to being degenerate, those cases in which the accuracy deteriorates to that of the naïve method may be attributed to the fact that the estimation of $\theta$ is not stable. Regardless of whether \eqref{linear} holds, if the propensity score model is estimated correctly, both the proposed method and the naive method can estimate the \MakeUppercase{att} without bias.

\begin{table}
\renewcommand{\baselinestretch}{1.5}\selectfont
\caption{Comparison of estimates of the \MakeUppercase{att} when the propensity scores are given by the \MakeUppercase{cbd} method and the \MakeUppercase{mle}.}
\begin{center}
\begin{tabular}{rrrrrrrrrr}
 & & & & \multicolumn{2}{c}{\MakeUppercase{cbd}-id} & \multicolumn{2}{c}{\MakeUppercase{cbd}-opt} & \multicolumn{2}{c}{\MakeUppercase{mle}} \\
 \multicolumn{1}{c}{$\beta^*$} & \multicolumn{1}{c}{$\alpha^*$} & \multicolumn{1}{c}{$n$} & \multicolumn{1}{c}{True} & \multicolumn{1}{c}{Est} & \multicolumn{1}{c}{Conf Int} & \multicolumn{1}{c}{Est} & \multicolumn{1}{c}{Conf Int} & \multicolumn{1}{c}{Est} & \multicolumn{1}{c}{Conf Int} \\ \addlinespace[1ex]
 & & 200 & 0.09 & 0.09 & [$-$0.12, 0.30] & 0.17 & [$-$0.20, 0.51] & $-$0.00 & [$-$0.22, 0.21] \\
            & 1.0 & 400 & 0.08 & 0.09 & [$-$0.07, 0.24] & 0.17 & [$-$0.14, 0.45]  & $-$0.01 & [$-$0.17, 0.14] \\
           \multirow{2}{*}{0.1} & & 600 & 0.08 & 0.08 & [$-$0.03, 0.21] & 0.19 & [$-$0.13, 0.12] & $-$0.01 & [$-$0.13, 0.12] \\ \addlinespace[1ex]
            & & 200 & 0.09 & 0.11 & [$-$0.15, 0.39] & 0.12 & [$-$0.14, 0.37] & $-$0.31 & [$-$0.80, 0.10] \\
            & 3.0 & 400 & 0.09 & 0.11 & [$-$0.07, 0.29] & 0.13 & [$-$0.06, 0.31] & $-$0.30 & [$-$0.60, $-$0.03] \\
            & & 600 & 0.10 & 0.11 & [$-$0.03, 0.26] & 0.12 & [$-$0.02, 0.27] & $-$0.30 & [$-$0.55, $-$0.08] \\ 
\addlinespace[1ex]
 & & 200 & 0.42 & 0.43 & [0.21, 0.65] & 0.52 & [0.09, 0.93] & 0.31 & [0.06, 0.57] \\
            & 1.0 & 400 & 0.43 & 0.43 & [0.27, 0.59] & 0.53 & [0.17, 0.85] & 0.31 & [0.15, 0.49] \\
           \multirow{2}{*}{0.5} & & 600 & 0.43 & 0.43 & [0.31, 0.56] & 0.55 & [0.20, 0.80] & 0.31 & [0.18, 0.45] \\ \addlinespace[1ex]
            & & 200 & 0.47 & 0.50 & [0.24, 0.78] & 0.51 & [0.25, 0.77] & 0.05 & [$-$0.45, 0.47] \\
            & 3.0 & 400 & 0.47 & 0.50 & [0.32, 0.69] & 0.52 & [0.33, 0.70] & 0.06 & [$-$0.32, 0.69] \\ 
            & & 600 & 0.48 & 0.50 & [0.35, 0.66] & 0.51 & [0.37, 0.66] & 0.05 & [$-$0.20, 0.28] \\
\addlinespace[1ex]
 & & 200 & 0.85 & 0.85 & [0.62, 1.10] & 0.97 & [0.47, 1.44] & 0.71 & [0.44, 1.01] \\
            & 1.0 & 400 & 0.85 & 0.85 & [0.68, 1.03] & 0.98 & [0.53, 1.35] & 0.70 & [0.52, 0.91] \\
           \multirow{2}{*}{1.0} & & 600 & 0.85 & 0.85 & [0.72, 0.99] & 1.00 & [0.57, 1.31] & 0.71 & [0.55, 0.87] \\ \addlinespace[1ex]
            & & 200 & 0.95 & 0.99 & [0.71, 1.28] & 1.00 & [0.73, 1.27] & 0.50 & [$-$0.02, 0.93] \\
            & 3.0 & 400 & 0.95 & 0.99 & [0.80, 1.18] & 1.00 & [0.81, 1.20] & 0.50 & [0.18, 0.80] \\
            & & 600 & 0.95 & 0.99 & [0.83, 1.15] & 1.00 & [0.85, 1.16] & 0.50 & [0.24, 0.74] \\
\addlinespace[1ex]
 & & 200 & 2.56 & 2.55 & [2.13, 2.96] & 2.77 & [1.84, 3.49] & 2.30 & [1.82, 2.82] \\
            & 1.0 & 400 & 2.56 & 2.55 & [2.27, 2.86] & 2.78 & [1.96, 2.67] & 2.30 & [1.95, 2.67] \\
           \multirow{2}{*}{3.0} & & 600 & 2.56 & 2.55 & [2.33, 2.76] & 2.82 & [2.03, 3.32] & 2.30 & [2.02, 2.58] \\ \addlinespace[1ex]
            & & 200 & 2.85 & 2.93 & [2.55, 3.33] & 2.95 & [2.59, 3.32] & 2.28 & [1.70, 2.85] \\
            & 3.0 & 400 & 2.85 & 2.92 & [2.67, 3.21] & 2.95 & [2.70, 3.21] & 2.28 & [1.89, 2.67] \\
            & & 600 & 2.85 & 2.92 & [2.71, 3.15] & 2.94 & [2.75, 3.16] & 2.27 & [1.97, 2.59] \\
\end{tabular}
\end{center}
True, true value of the \MakeUppercase{att}; \MakeUppercase{cbd}-id, \MakeUppercase{cbd} method when the weighting matrix in the generalized method of moments is the identity matrix; \MakeUppercase{cbd}-opt, \MakeUppercase{cbd} method when it is the optimal matrix; Est, average of the estimates obtained by each simulation; Conf Int, empirical 95\% confidence interval.
\label{tab4}
\end{table}

\section{Model selection criterion for semiparametric difference-in-differences when propensity scores are known}
\label{sec3}
\subsection{Risk function and its asymptotic evaluation in semiparametric difference-in-differences}
\label{sec4_1}
As mentioned in Section \ref{sec0}, there are no reasonable model selection criteria even for the basic setting treated in {\cite{abadie2005}. To remedy this situation, in this section, we will first derive a model selection criterion for when the propensity scores are known and then develop a model selection criterion for when the proposed \MakeUppercase{cbd} method is used. Specifically, by using the ideas of \cite{BabKN17} and \cite{BabNino}, we will derive the model selection criteria as asymptotically unbiased estimators of the natural risk corresponding to the loss function used to obtain \MakeUppercase{sdid} estimators. Using these criteria, the dimension of $\theta$ is determined and the component of $x$ is selected. A natural choice is to define a weighted risk based on \eqref{abat0} as
\begin{align}
& \sum_{i=1}^{n}\E[e^{[1]}(x_i;\alpha)\{\E[y^{[1]}_{i}(1)-y^{[0]}_{i}(1)\mid x_i,d^{[1]}_{i}=1]-x_{i}^{\T}\hat{\theta}\}^2] \notag \\
& = \sum_{i=1}^{n}\E[e^{[1]}(x_i;\alpha)\{\rho(d_i,x_i;\alpha)\Delta_i-x_{i}^{\T}\hat{\theta}\}^2] \notag \\
& \ \phantom{=} -\sum_{i=1}^{n}\E[e^{[1]}(x_i;\alpha)\{\rho(d_i,x_i;\alpha)\Delta_i-\E[y^{[1]}_{i}(1)-y^{[0]}_{i}(1)\mid x_i,d^{[1]}_{i}=1]\}^2] \notag
\\
& \ \phantom{=} +2\sum_{i=1}^{n}\E[e^{[1]}(x_i;\alpha)\{\rho(d_i,x_i;\alpha)\Delta_i-\E[y^{[1]}_{i}(1)-y^{[0]}_{i}(1)\mid x_i,d^{[1]}_{i}=1]\} \notag \\
& \ \phantom{=+2\sum_{i=1}^{n}\E[} \times \{x_{i}^{\T}\hat{\theta}-\E[y^{[1]}_{i}(1)-y^{[0]}_{i}(1)\mid x_i,d^{[1]}_{i}=1]\}].
\label{risk}
\end{align}
Following the derivation of the conventional \MakeUppercase{c}$_{\text{\MakeUppercase{p}}}$ criterion, we will remove the expectation in the first term, ignore the second term which is independent of the models, and asymptotically evaluate the third term. For simplicity, we will refer to the third term as a bias when the risk is simply estimated. If we define $\varepsilon_i \equiv \rho(d_i,x_i;\alpha)\Delta_i-a(x_i)$ using $a(x_i)$ in \eqref{ATTx}, the third term can be written as
\begin{align}
& 2\sum_{i=1}^n\E[e^{[1]}(x_i;\alpha)\{\rho(d_i,x_i;\alpha)\Delta_i-a(x_i)\}\{x_{i}^{\T}\hat{\theta}-a(x_i)\}] \notag \\
& = 2\sum_{i=1}^n\E[\varepsilon_ie^{[1]}(x_i;\alpha)\{x_{i}^{\T}\hat{\theta}-a(x_i)\}] 
 = 2\sum_{i=1}^n\E[\varepsilon_ie^{[1]}(x_i;\alpha)x_{i}^{\T}(\hat{\theta}-\theta^*)].
\label{pena}
\end{align}
The second equality is satisfied by $\E[\varepsilon_i e^{[1]}(x_i;\alpha) a(x_i)]=0$ and $\E[\varepsilon_i e^{[1]}(x_i;\alpha) \allowbreak x^{\T}_{i}\theta^{*}]=0$ since it follows from \eqref{abaprop} that $\E[\varepsilon_i\mid x_i]=0$. Here, we have used the same $e^{[1]}(x_i;\alpha)$ as that was used to construct $\hat{\theta}$. Then, letting $c^{\rm limit}$ be the weak limit of the third term with the expectation removed, we consider $\E[c^{\rm limit}]$ as the asymptotic evaluation of the third term of \eqref{risk}. While we consider a risk based on the loss function in deriving the \MakeUppercase{sdid} estimator, it is also possible to consider a simple squared error.

\subsection{Derivation of model selection criterion}
\label{sec4_2}
In the case of known propensity scores, we will denote $e^{[1]}(x_i;\alpha^*)$ as $e^{[1]}(x_i)$ and $\rho(d_i,x_i;\alpha^*)$ as $\rho(d_i,x_i)$ for simplicity. The main term of the error of the \MakeUppercase{sdid} estimator can be expressed as
\begin{align*}
& \hat{\theta}-\theta^* = \E[e^{[1]}(x)xx^{\T}]^{-1}\frac{1}{n}\sum_{i=1}^{n}e^{[1]}(x_i)x_{i}\{\rho(d_i,x_i)\Delta_i-x_i^{\T}\theta^*\}\{1+\oP(1)\}
\\
& = \E[e^{[1]}(x)xx^{\T}]^{-1}\frac{1}{n}\sum_{i=1}^{n}e^{[1]}(x_i)x_{i}\{\varepsilon_i+a(x_i)-x_i^{\T}\theta^*\}\{1+\oP(1)\}.
\end{align*}
Accordingly, \eqref{pena} is evaluated as
\begin{align*}
\E[c^{\rm limit}] = 2\sum_{i=1}^{n}\E\bigg[\varepsilon_i e^{[1]}(x_i)x_{i}^{\T}\E[e^{[1]}(x)xx^{\T}]^{-1}\frac{1}{n}\sum_{j=1}^{n}e^{[1]}(x_j)x_{j}\{\varepsilon_j+a(x_j)-x_j^{\T}\theta^*\}\bigg].
\end{align*}
This expectation for the case of $i\neq j$ is the product of the expectations for $i$ and $j$, from the independence among the samples. The expectation for $i$ can be calculated as 
\begin{align}
\E[\varepsilon_ie^{[1]}(x_i)x_{i}^{\T}] = \E_{x_{i}}[\E[\varepsilon_i\mid x_{i}]e^{[1]}(x_i)x_{i}^{\T}] = 0.
\label{zeroterm}
\end{align}
Therefore, we only have to consider the case of $i=j$. Since $a(x_j)-x_j^{\T}\theta^*$ contains only $x_j$ as a random variable, this term disappears from \eqref{zeroterm}, and the following theorem is obtained.

\begin{theorem}\label{theorem2}
When the propensity scores are known, it holds that
\begin{align*}
& \E[c^{\rm limit}] =2\E[\varepsilon e^{[1]}(x)x^{\T}\E[e^{[1]}(x)xx^{\T}]^{-1}e^{[1]}(x)x\varepsilon] \\
& = 2{\rm tr}\{\E[e^{[1]}(x)xx^{\T}]^{-1}\E[\sigma^2(x)e^{[1]}(x)^2xx^{\T}]\},
\end{align*}
where
\begin{align}
\sigma^2(x) \equiv \E[\varepsilon^2\mid x]
& = \E[\{\rho(d,x)\Delta\}^2\mid x] - \E[\rho(d,x)\Delta\mid x]^2. \label{std}
\end{align}
\end{theorem}

\noindent
Although one may think that an empirical evaluation of $\E[c^{\rm limit}]$ is easy, we must consider how to deal with the term $\E[\E[\rho(d,x)\Delta\mid x]^2e^{[1]}(x)^2 xx^{\T}]$. This cannot be empirically evaluated as $n^{-1}\sum_{i=1}^{n}\rho(d,x)^2\allowbreak\Delta^2e^{[1]}(x)^2 xx^{\T}$. Here, we will substitute $x_i^{\T}\hat{\theta}$ into $\E[\rho(d_i,x_i)\Delta\mid x_i]$. As a result, when the propensity scores are known, we propose
\begin{align}
&\sum_{i=1}^{n}e^{[1]}(x_i)\{\rho(d_i,x_i)\Delta_i-x_{i}^{\T}\hat{\theta}\}^2 \notag \\
&+2{\rm tr}\biggl(\biggl\{\sum_{i=1}^ne^{[1]}(x_i)x_ix_i^{\T}\biggr\}^{-1}\biggl[\sum_{i=1}^n\{\rho(d_i,x_i)^2\Delta^2-(x_i^{\T}\hat{\theta})^2\}e^{[1]}(x_i)^2x_ix_i^{\T}\biggr]\biggr)
\label{IC_k}
\end{align}
as the model selection criterion for the basic \MakeUppercase{sdid} approach.

Through simulation studies, we demonstrate that the second term in \eqref{IC_k}, which serves as the penalty term of the proposed criterion, accurately approximates the bias in \eqref{pena} that we originally wanted to evaluate.
In contrast, the penalty term of an intuitive model selection criterion obtained by extending \MakeUppercase{qic}$_{\text{\MakeUppercase{w}}}$ (\citealt{PlaBCWS13}) to the present setting substantially underestimates this bias.
We further compare the two criteria using an empirical estimate of the risk as the primary performance measure. The results show that the proposed criterion outperforms \MakeUppercase{qic}$_{\text{\MakeUppercase{w}}}$ in almost all settings.
Details on these numerical experiments, including the extension of \MakeUppercase{qic}$_{\text{\MakeUppercase{w}}}$, are provided in the Supplementary Material.
 
\section{Model selection criterion for covariate balancing difference-in-differences}
\label{sec5}

\subsection{Derivation of model selection criterion}
\label{sec5_1}

As we did for the case of known propensity scores in Section \ref{sec4_2}, here, we will derive a validated model selection criterion for when the \MakeUppercase{cbd} method is used. Let us evaluate
\begin{align*}
\sum_{i=1}^n 2\E[e^{[1]}(x_i;\hat{\alpha}^{\rm \MakeUppercase{cb}})\{\rho(d_i,x_i;\hat{\alpha}^{\rm \MakeUppercase{cb}})\Delta_i-\E[y^{[1]}_{i}-y^{[0]}_{i} \mid x_i, d^{[1]}_{i}=1]\}x_{i}^{\T}(\hat{\theta}^{\rm \MakeUppercase{cbd}}-\theta^*)],
\end{align*}
which is the third term of the risk in \eqref{pena}. Substituting $x_i^{\T}\theta^*$, which is a working model for conditional \MakeUppercase{att}, into $\E[y^{[1]}_{i}-y^{[0]}_{i} \mid x_i, d^{[1]}_{i}=1]$ first and then substituting the right-hand side of \eqref{expand} into $\hat{\theta}^{\rm \MakeUppercase{cbd}}-\theta^*$ yields
\begin{align*}
& 2\E\bigg[\frac{1}{\sqrt{\vphantom{x}}{n}}\sum_{i=1}^n e^{[1]}(x_{i};\hat{\alpha}^{\rm \MakeUppercase{cb}})\{\rho(d_i,x_i;\hat{\alpha}^{\rm \MakeUppercase{cb}})\Delta_i-x_i^{\T}\theta^*\}x_i^{\T} \bigg\{\frac{1}{n}\sum_{j=1}^ne^{[1]}(x_j;\hat{\alpha}^{\rm \MakeUppercase{cb}})x_jx_j^{\T}\bigg\}^{-1} 
\\
& \ \phantom{\E\bigg[} \times\frac{1}{\sqrt{\vphantom{x}}{n}}\sum_{j=1}^n e^{[1]}(x_j;\hat{\alpha}^{\rm \MakeUppercase{cb}})\{\rho(d_j,x_j;\hat{\alpha}^{\rm \MakeUppercase{cb}})\Delta_j-x_j^{\T}\theta^*\}x_j+\oP(1)\bigg].
\end{align*}
Since $n^{-1}\sum_{j=1}^ne^{[1]}(x_j;\hat{\alpha}^{\rm \MakeUppercase{cb}})x_jx_j^{\T}$ converges in probability to
\begin{align}
L(\alpha^{\dagger})\equiv\E[e^{[1]}(x;\alpha^{\dagger})xx^{\T}],
\label{defL}
\end{align}
by letting $u$ be the Gaussian random vector that is the weak limit of
\begin{align}
\frac{1}{\sqrt{\vphantom{x}}{n}}\sum_{i=1}^ne^{[1]}(x_i;\hat{\alpha}^{\rm \MakeUppercase{cb}})\{\rho(d_i,x_i;\hat{\alpha}^{\rm \MakeUppercase{cb}})\Delta_i-x_i^{\T}\theta^*\}x_i,
\label{weak}
\end{align}
we can express the weak limit of \eqref{pena} with expectation removed as
\begin{align*}
c^{\rm limit} = 2u^{\T}L(\alpha^{\dagger})^{-1}u = 2{\rm tr}\{L(\alpha^{\dagger})^{-1}uu^{\T}\}.
\end{align*}
When $\hat{\theta}^{\rm \MakeUppercase{cbd}}$ is a consistent estimator of $\theta^*$, it holds that ${\rm E}[uu^{\T}]={\rm var}[u]$ since the expectation of $u$ becomes $0$. Accordingly, the asymptotic evaluation of \eqref{pena}, which is to be used in the information criterion, can be written as
\begin{align*}
\E[c^{\rm limit}] =  2{\rm tr}\{L(\alpha^{\dagger})^{-1}{\rm var}[u]\}.
\end{align*}

Now let us derive a specific expression for ${\rm var}[u]$. Since the error of the estimator $\hat{\alpha}^{\rm \MakeUppercase{cb}}$ in the propensity scores is written as \eqref{adiff}, we can expand \eqref{weak} as
\begin{align*}
& \frac{1}{\sqrt{\vphantom{x}}{n}}\sum_{i=1}^ne^{[1]}(x_i;\hat{\alpha}^{\rm \MakeUppercase{cb}})\{\rho(d_i,x_i;\hat{\alpha}^{\rm \MakeUppercase{cb}})\Delta_i-x_i^{\T}\theta^*\}x_i \\
& = \frac{1}{\sqrt{\vphantom{x}}{n}}\sum_{i=1}^{n}x_i\bigg\{e^{[1]}(x_i;\alpha^{\dagger})\{\rho(d_i,x_i;\alpha^{\dagger})\Delta_i-x_i^{\T}\theta^*\} \\
& \ \phantom{= \frac{1}{\sqrt{\vphantom{x}}{n}}\sum_{i=1}^{n}x_i\bigg\{} + \dfrac{\partial}{\partial \alpha^{\T}}\bigg(e^{[1]}(x_i;\alpha^{\dagger})\bigg[\frac{d^{[1]}_{i}-e^{[1]}(x_i;\alpha^{\dagger})}{e^{[1]}(x_i;\alpha^{\dagger})\{1-e^{[1]}(x_i;\alpha^{\dagger})\}}\Delta_i-x_i^{\T}\theta^*\bigg]\bigg)(\hat{\alpha}^{\rm \MakeUppercase{cb}}-\alpha^{\dagger}) \bigg\}
\\
& \ \phantom{= \frac{1}{\sqrt{\vphantom{x}}{n}}\sum_{i=1}^{n}x_i} \times\{1+\oP(1)\} \\
& = \frac{1}{\sqrt{\vphantom{x}}{n}}\sum_{i=1}^{n}[e^{[1]}(x_i;\alpha^{\dagger})\{\rho(d_i,x_i;\alpha^{\dagger})\Delta_i-x_i^{\T}\theta^*\}x_{i} \\
& \ \phantom{= \frac{1}{\sqrt{\vphantom{x}}{n}}\sum_{i=1}^{n}[} - M(\alpha^{\dagger},\theta^{*})\{G(\alpha^{\dagger})^{\T}WG(\alpha^{\dagger})\}^{-1}G(\alpha^{\dagger})^{\T}Wh(d_i,x_i;\alpha^{\dagger})]\{1+\oP(1)\},
\end{align*}
where
\begin{align*}
M(\alpha,\theta)=\E\biggl[\dfrac{\partial}{\partial \alpha}e^{[1]}(x;\alpha)\biggl[\frac{d^{[1]} -1}{\{1-e^{[1]}(x;\alpha)\}^2}\Delta - x^{\T}\theta\biggr]x^{\T}\biggr].
\end{align*}
Then, by defining 
\begin{align}
V(\alpha^{\dagger},\theta^*) \equiv {\rm var} & [e^{[1]}(x;\alpha^{\dagger})\{\rho(d,x;\alpha^{\dagger})\Delta-x^{\T}\theta^*\}x \notag \\
& \; -M(\alpha^{\dagger},\theta^{*})\{G(\alpha^{\dagger})^{\T}WG(\alpha^{\dagger})\}^{-1}G(\alpha^{\dagger})^{\T}Wh(d,x;\alpha^{\dagger})],
\label{vari}
\end{align}
it can be shown that ${\rm var}[u]=V(\alpha^{\dagger},\theta^*)$. Thus, the following theorem is obtained.

\begin{theorem}\label{theorem3}
When the propensity scores are estimated using the \MakeUppercase{cbd} method, the third term in \eqref{risk} is evaluated as
\begin{align*}
\E[c^{\rm limit}] =  2{\rm tr}\{L(\alpha^{\dagger})^{-1}V(\alpha^{\dagger},\theta^*)\},
\end{align*}
where $L(\cdot)$ and $V(\cdot,\cdot)$ are defined in \eqref{defL} and \eqref{vari}, respectively. 
\end{theorem}

Since the penalty term in this theorem depends on the true values of the parameters, we can use empirical estimates in practice. 
Specifically, letting 
\begin{align}
M_n(\alpha,\theta)\equiv \frac{1}{n}\sum_{i=1}^n\dfrac{\partial}{\partial \alpha}e^{[1]}(x_i;\alpha)\bigg[\frac{d^{[1]}_{i}-1}{\{1-e^{[1]}(x_i;\alpha)\}^2}\Delta_i - x_i^{\T}\theta\bigg]x_i^{\T}
\label{defMn}
\end{align}
and
\begin{align*}
V_i(\alpha,\theta) \equiv \ & e^{[1]}(x_i;\alpha)\{\rho(d_i,x_i;\alpha)\Delta_i-x_i^{\T}\theta\}x_i \\
& - M_n(\alpha,\theta)\{G_n(\alpha)^{\T}W_{n}G_n(\alpha)\}^{-1}G_n(\alpha)^{\T}W_{n}h(d_i,x_i;\alpha)],
\end{align*}
we will replace $L(\alpha^{\dagger})$ by $L_n(\hat{\alpha}^{\rm \MakeUppercase{cb}}) \equiv n^{-1}\sum_{i=1}^n\allowbreak e^{[1]}(x_i;\hat{\alpha}^{\rm \MakeUppercase{cb}})x_ix_i^{\T}$ and $V(\alpha^{\dagger},\theta^*)$ by $V_n(\hat{\alpha}^{\rm \MakeUppercase{cb}},\hat{\theta}^{\rm \MakeUppercase{cbd}}) \equiv n^{-1}\sum_{i=1}^n V_i(\hat{\alpha}^{\rm \MakeUppercase{cb}},\hat{\theta}^{\rm \MakeUppercase{cbd}})V_i(\hat{\alpha}^{\rm \MakeUppercase{cb}},\hat{\theta}^{\rm \MakeUppercase{cbd}})^{\T}$. 
As a result, for when the propensity scores are estimated with the \MakeUppercase{cbd} method, we propose
\begin{align}
\sum_{i=1}^{n}e^{[1]}(x_{i};\hat{\alpha}^{\rm \MakeUppercase{cb}})\{\rho(d_i,x_i;\hat{\alpha}^{\rm \MakeUppercase{cb}})\Delta_i-x_{i}^{\T}\hat{\theta}^{\rm \MakeUppercase{cbd}}\}^2 + 2{\rm tr}\{L_n(\hat{\alpha}^{\rm \MakeUppercase{cb}})^{-1}V_n(\hat{\alpha}^{\rm \MakeUppercase{cb}},\hat{\theta}^{\rm \MakeUppercase{cbd}})\}
\label{CBDIC}
\end{align}
as the model selection criterion, which is an asymptotically unbiased estimator of the risk in \eqref{risk}. Here, the propensity scores are estimated using covariate balancing, but the information criterion can be derived in the same way when the \MakeUppercase{mle} is used (see the Supplementary Material).

\subsection{Numerical experiments}
\label{sec5_2}
The simulation studies presented in this section examine the accuracy of the approximation for the penalty term and the performance of the model selection criterion in \eqref{CBDIC} for the proposed \MakeUppercase{cbd} method. On the basis of the results in Section \ref{robust}, we will use the identity matrix as the weighting matrix in the generalized method of moments and, even when the optimal matrix is used as the weighting matrix, the accuracy of the approximation for the penalty term and the performance of the model selection with comparable to those obtained with the identity matrix (see the Supplementary Material).

In Table \ref{tab5}, the data generation settings include Case 1-1,
\begin{align*}
& x_1 \sim{\rm Uniform}(0, 2), \quad d^{[1]}\sim{\rm Bernoulli}({\rm logit}(-x_1)), \quad y^{[0]}(0)=y^{[1]}(0)\sim{\rm N}(0,1),
\\
& \epsilon^{[0]}, \epsilon^{[1]} \sim{\rm N}(0,1), \quad y^{[0]}(1)=y^{[0]}(0)+\epsilon^{[0]}, \quad y^{[1]}(1)=y^{[1]}(0)+1+\beta^* x_1+\epsilon^{[1]},
\end{align*}
and Case 1-2,
\begin{align*}
& x_1,x_2 \sim{\rm Uniform}(0, 2), \quad d^{[1]}\sim{\rm Bernoulli}({\rm logit}(-x_1+x_2)), \quad y^{[0]}(0)=y^{[1]}(0)\sim{\rm N}(0,1), 
\\
& \epsilon^{[0]}, \epsilon^{[1]} \sim{\rm N}(0,1), \quad y^{[0]}(1)=y^{[0]}(0)+\epsilon^{[0]}, \quad y^{[1]}(1)=y^{[1]}(0)+1+\beta^* (x_1+x_2)+\epsilon^{[1]}.
\end{align*}
In both cases, $\beta^*$ represents the magnitude of contribution of the covariate $x_1$ to the causal effect.
Results are listed to check whether the second term in \eqref{CBDIC}, the penalty term of the proposed criterion, accurately approximates the bias in \eqref{pena}, which we originally wanted to evaluate, for various values of the parameter $\beta^*$ and sample size $n$. In all cases, the evaluated penalty terms of the proposed criterion are much closer to the bias than the evaluated penalty terms of the \MakeUppercase{qic}$_{\text{\MakeUppercase{w}}}$, and its approximation accuracy is high. In contrast, the penalty term of the \MakeUppercase{qic}$_{\text{\MakeUppercase{w}}}$ significantly underestimates the bias in all cases.

\begin{table}[!t]
\renewcommand{\baselinestretch}{1.5}\selectfont
\caption{Bias evaluation for case in which the propensity scores are estimated using the \MakeUppercase{cbd} method with the identity matrix as the weighting matrix.}
\begin{center}
\begin{tabular}{rrrrrrrr}
&  & \multicolumn{3}{c}{Case 1-1} & \multicolumn{3}{c}{Case 1-2} \\
 \multicolumn{1}{c}{$\beta^*$} & \multicolumn{1}{c}{$n$} & \multicolumn{1}{c}{True} & \multicolumn{1}{c}{Proposal} & \multicolumn{1}{c}{\MakeUppercase{qic}$_{\text{\MakeUppercase{w}}}$} & \multicolumn{1}{c}{True} & \multicolumn{1}{c}{Proposal} & \multicolumn{1}{c}{\MakeUppercase{qic}$_{\text{\MakeUppercase{w}}}$} \\ \addlinespace[1ex]
 & 200 & ~7.53 & ~7.35 & 2.23 & 17.61 & 16.98 & ~5.94 \\
 0.1 & 400 & ~7.19 & ~7.31 & 2.24 & 17.88 & 17.55 & ~5.99 \\
 & 600 & ~7.09 & ~7.34 & 2.27 & 17.10 & 17.35 & ~5.98 \\ \addlinespace[1ex]
 & 200 & ~8.21 & ~8.04 & 2.31 & 19.42 & 18.76 & ~6.38 \\
 0.5 & 400 & ~7.79 & ~7.89 & 2.33 & 19.51 & 19.62 & ~6.48 \\
 & 600 & ~7.75 & ~7.88 & 2.34 & 18.94 & 19.25 & ~6.47 \\ \addlinespace[1ex]
 & 200 & ~9.47 & ~9.56 & 2.58 & 23.20 & 22.67 & ~7.92 \\
 1.0 & 400 & ~9.04 & ~9.25 & 2.58 & 23.07 & 23.18 & ~8.00 \\
 & 600 & ~8.96 & ~9.14 & 2.59 & 22.62 & 22.98 & ~8.00 \\ \addlinespace[1ex]
 & 200 & 19.10 & 19.22 & 5.17 & 54.55 & 54.82 & 23.62 \\
 3.0 & 400 & 18.36 & 18.58 & 5.20 & 53.71 & 54.81 & 24.08 \\
 & 600 & 17.88 & 18.49 & 5.22 & 53.55 & 54.64 & 23.90 \\
\end{tabular}
\end{center}
True, the true value of the bias evaluated by the Monte Carlo method; Proposal, the penalty term of the proposed criterion; \MakeUppercase{qic}$_{\text{\MakeUppercase{w}}}$, the penalty term of the \MakeUppercase{qic}$_{\text{\MakeUppercase{w}}}$.
\label{tab5}
\end{table}

In Table \ref{tab6}, data are generated from
\begin{align*}
& x_1,\ldots,x_l\sim{\rm Uniform}(0, 2), \quad x=(1,x_1,\ldots,x_{l})^{\T}, \quad y^{[0]}(0)=y^{[1]}(0)\sim{\rm N}(0,1), 
\\
& \epsilon^{[0]},\epsilon^{[1]} \sim{\rm N}(0,1), \quad y^{[0]}(1)=y^{[0]}(0)+\epsilon^{[0]}, \quad y^{[1]}(1)=y^{[1]}(0)+x^{\T}\beta^{[1]*}+\epsilon^{[1]},
\end{align*}
in Case 2-1,
\begin{align*}
l=4, \quad \beta^{[1]*} = (1,\beta^*,0_{3}^{\T})^{\T}, \quad d^{[1]}\sim{\rm Bernoulli}({\rm logit}(-x_1)),
\end{align*}
Case 2-2,
\begin{align*}
l=4, \quad \beta^{[1]*} = (1,1_2^{\T}\beta^*,0_{2}^{\T})^{\T}, \quad d^{[1]}\sim{\rm Bernoulli}({\rm logit}(-x_1+x_2)),
\end{align*}
and Case 2-3, 
\begin{align*}
l=6, \quad \beta^{[1]*} = (1,1_2^{\T}\beta^*,0_{4}^{\T})^{\T}, \quad d^{[1]}\sim{\rm Bernoulli}({\rm logit}(-x_1+x_2)).
\end{align*}
Results are listed to compare the proposed criterion and the \MakeUppercase{qic}$_{\text{\MakeUppercase{w}}}$ through the forward selection method. The main index of goodness for the model selection criteria is again the empirical estimate of the risk. Unlike Section \ref{sec4_4}, where the propensity scores are known, the proposed criterion always outperforms the \MakeUppercase{qic}$_{\text{\MakeUppercase{w}}}$ in all cases, even in Case 2-2, where relatively few covariates are unnecessary among the candidates. The proposed criterion is especially superior when the candidates include many covariates that do not contribute to the true structure, as in Case 2-1 and Case 2-3 of Section \ref{sec4_4}. The \MakeUppercase{qic}$_{\text{\MakeUppercase{w}}}$ underestimates the penalty term; thus, too many covariates are selected.

\begin{table}
\renewcommand{\baselinestretch}{1.5}\selectfont
\caption{Comparison of the proposed criterion and the \MakeUppercase{qic}$_{\text{\MakeUppercase{w}}}$ through model selection for case in which the propensity scores are estimated using the \MakeUppercase{cbd} method with the identity matrix as the weighting matrix.}
\begin{center}
\begin{tabular}{cccrrrrrrrrr}
& & & \multicolumn{3}{c}{Case 2-1} & \multicolumn{3}{c}{Case 2-2} & \multicolumn{3}{c}{Case 2-3} \\
$\beta^*$ & $n$ & & \multicolumn{1}{c}{Risk} & \multicolumn{1}{c}{\MakeUppercase{tp}} & \multicolumn{1}{c}{\MakeUppercase{fp}} & \multicolumn{1}{c}{Risk} & \multicolumn{1}{c}{\MakeUppercase{tp}} & \multicolumn{1}{c}{\MakeUppercase{fp}} & \multicolumn{1}{c}{Risk} & \multicolumn{1}{c}{\MakeUppercase{tp}} & \multicolumn{1}{c}{\MakeUppercase{fp}} \\ \addlinespace[1ex]
& \multirow{2}{*}{200} & Proposal & 7.33 & 0.21 & 0.53 & 12.38 & 0.41 & 0.33 & 16.49 & 0.38 & 0.74 \\
& & \MakeUppercase{qic}$_{\text{\MakeUppercase{w}}}$ & 10.32 & 0.48 & 1.55 & 16.96 & 0.91 & 0.92 & 22.57 & 0.84 & 1.79 \\ \addlinespace[1ex]
\multirow{2}{*}{0.1} & \multirow{2}{*}{400} & Proposal & 7.30 & 0.21 & 0.49 & 12.06 & 0.43 & 0.33 & 17.02 & 0.46 & 0.69 \\
& & \MakeUppercase{qic}$_{\text{\MakeUppercase{w}}}$ & 10.13 & 0.47 & 1.46 & 15.75 & 0.96 & 0.91 & 22.47 & 0.96 & 1.74 \\ \addlinespace[1ex]
& \multirow{2}{*}{600} & Proposal & 6.88 & 0.20 & 0.46 & 12.01 & 0.46 & 0.34 & 16.62 & 0.50 & 0.66 \\
& & \MakeUppercase{qic}$_{\text{\MakeUppercase{w}}}$ & 9.57 & 0.50 & 1.43 & 15.11 & 0.96 & 0.91 & 21.91 & 0.98 & 1.79 \\ \addlinespace[1ex]
& \multirow{2}{*}{200} & Proposal & 10.73 & 0.51 & 0.55 & 20.12 & 1.12 & 0.34 & 26.11 & 1.11 & 0.72 \\
& & \MakeUppercase{qic}$_{\text{\MakeUppercase{w}}}$ & 13.12 & 0.77 & 1.65 & 21.56 & 1.55 & 1.01 & 30.24 & 1.51 & 2.01 \\ \addlinespace[1ex]
\multirow{2}{*}{0.5} & \multirow{2}{*}{400} & Proposal & 10.64 & 0.72 & 0.50 & 20.34 & 1.52 & 0.33 & 26.36 & 1.54 & 0.69 \\
& & \MakeUppercase{qic}$_{\text{\MakeUppercase{w}}}$ & 12.87 & 0.88 & 1.61 & 20.47 & 1.80 & 1.02 & 29.42 & 1.80 & 1.97 \\ \addlinespace[1ex]
& \multirow{2}{*}{600} & Proposal & 10.02 & 0.84 & 0.48 & 17.98 & 1.79 & 0.33 & 24.29 & 1.77 & 0.66. \\
& & \MakeUppercase{qic}$_{\text{\MakeUppercase{w}}}$ & 12.05 & 0.97 & 1.60 & 19.15 & 1.92 & 1.03 & 28.46 & 1.92 & 2.01 \\ \addlinespace[1ex]
& \multirow{2}{*}{200} & Proposal & 13.83 & 0.86 & 0.54 & 25.04 & 1.78 & 0.33 & 34.29 & 1.80 & 0.72 \\
 & & \MakeUppercase{qic}$_{\text{\MakeUppercase{w}}}$ & 17.55 & 0.95 & 1.78 & 28.00 & 1.92 & 1.10 & 41.46 & 1.93 & 2.15 \\ \addlinespace[1ex]
 \multirow{2}{*}{1.0} & \multirow{2}{*}{400} & Proposal & 12.44 & 0.98 & 0.50 & 22.03 & 1.97 & 0.33 & 31.09 & 1.97 & 0.67 \\
 & & \MakeUppercase{qic}$_{\text{\MakeUppercase{w}}}$ & 16.82 & 1.00 & 1.76 & 26.42 & 2.00 & 1.10 & 40.19 & 2.00 & 2.16 \\ \addlinespace[1ex]
 & \multirow{2}{*}{600} & Proposal & 11.33 & 1.00 & 0.48 & 20.06 & 2.00 & 0.31 & 28.79 & 2.00 & 0.63 \\
 & & \MakeUppercase{qic}$_{\text{\MakeUppercase{w}}}$ & 15.91 & 1.00 & 1.76 & 25.28 & 2.00 & 1.11 & 39.30 & 2.00 & 2.17 \\ \addlinespace[1ex]
 & \multirow{2}{*}{200} & Proposal & 34.16 & 1.00 & 0.52 & 60.56 & 2.00 & 0.31 & 95.78 & 2.00 & 0.68 \\
 & & \MakeUppercase{qic}$_{\text{\MakeUppercase{w}}}$ & 48.77 & 1.00 & 1.94 & 81.22 & 2.00 & 1.17 & 136.38 & 2.00 & 2.37 \\ \addlinespace[1ex]
 \multirow{2}{*}{3.0} & \multirow{2}{*}{400} & Proposal & 31.96 & 1.00 & 0.51 & 59.75 & 2.00 & 0.31 & 91.38 & 2.00 & 0.65 \\
 & & \MakeUppercase{qic}$_{\text{\MakeUppercase{w}}}$ & 46.25 & 1.00 & 1.95 & 80.71 & 2.00 & 1.18 & 132.43 & 2.00 & 2.36 \\ \addlinespace[1ex]
 & \multirow{2}{*}{600} & Proposal & 29.56 & 1.00 & 0.49 & 57.25 & 2.00 & 0.31 & 87.98 & 2.00 & 0.61 \\
 & & \MakeUppercase{qic}$_{\text{\MakeUppercase{w}}}$ & 44.02 & 1.00 & 1.95 & 78.45 & 2.00 & 1.18 & 130.44 & 2.00 & 2.34 \\
\end{tabular}
\end{center}
Risk, the empirical estimate of the risk of the selected model; \MakeUppercase{tp}, the average number of covariates that are true positive; \MakeUppercase{fp}, the average number of covariates that are false positive.
\label{tab6}
\end{table}

\section{Real data analysis}
\label{sec6}
The LaLonde dataset, originally presented in \cite{Lal86}, is included in the R package Matching version 4.10.14. 
In the dataset, the group that took the U.S. job training program in 1976 is denoted as $t=1$ and the group that did not take the program is denoted as $t=0$. The difference in annual income in 1978 after the training for each group is estimated via the \MakeUppercase{att}. Age (age}, years of education (educ), income in 1974 (re74), black (black), hispanic (hisp), married (married), and high school graduate or higher (nodegr) are covariates, i.e. $p=7$, and the outcome variable is the difference between income in 1978 (re78) and income in 1974 (re74). The sample size is $n=445$. The objective of our experiment here is not to extract information from these data, but to determine whether the proposed method makes a significant difference from an existing method. Therefore, to reduce the sample size of the dataset, we divided the data into three parts and analyzed each of them. Specifically, for every integer $k$, the $3k-2$-th row of data are included in Block 1, the $3k-1$-th row in Block 2, and the $3k$-th row in Block 3, respectively.

First, we estimate the propensity scores by using the covariate balancing with the identity matrix as the weighting matrix. Then, supposing that the \MakeUppercase{att} can be written as a linear sum of covariates, we estimate their coefficients by using the \MakeUppercase{sdid} approach. In Table \ref{tab7}, the covariates are selected using the proposed criterion or the \MakeUppercase{qic}$_{\text{\MakeUppercase{w}}}$ with the forward selection method and the coefficients of the selected covariates are estimated. The \MakeUppercase{qic}$_{\text{\MakeUppercase{w}}}$ selects all covariates in any block, while the proposed criterion does not select some variables. In particular, the difference between the models selected for Block 1 and Block 3 is rather large. Because we do not know the true structure of the real data, it is impossible to judge which of the two criteria is best; however, we can see that the difference between them is quite large. These results indicates that, in a setting such as here, it is essential to use a model selection criterion that is theoretically valid rather than an intuitive one.

\begin{table}[!t]
\renewcommand{\baselinestretch}{1.5}\selectfont
\caption{Comparison of the proposed criterion and the \MakeUppercase{qic}$_{\text{\MakeUppercase{w}}}$ through model selection for the LaLonde dataset. The numbers with exponential notation are the estimated coefficients for each covariate.}
\begin{center}
\begin{tabular}{ccrrrr}
Block & & \multicolumn{1}{c}{intercept} & \multicolumn{1}{c}{age} & \multicolumn{1}{c}{educ} & \multicolumn{1}{c}{re74} \\ \addlinespace[1ex]
 \multirow{2}{*}{1} & Proposal & $-$2.14$_{\,\E+3}$ & 0\phantom{{}$_{\,\E+0}$} & 0\phantom{{}$_{\,\E+0}$} & 0\phantom{{}$_{\,\E+0}$} \\
 & \MakeUppercase{qic}$_{\text{\MakeUppercase{w}}}$ & $-$1.37$_{\,\E+4}$ & 5.29$_{\,\E+2}$ & 4.43$_{\,\E+1}$ & $-$2.01$_{\,\E-1}$ \\ \addlinespace[1ex]
 \multirow{2}{*}{2} & Proposal & 1.59$_{\,\E+4}$ & 0\phantom{{}$_{\,\E+0}$} & 1.08$_{\,\E+3}$ & 2.54$_{\,\E-1}$ \\
 & \MakeUppercase{qic}$_{\text{\MakeUppercase{w}}}$ & 2.38$_{\,\E+4}$ & $-$1.91$_{\,\E+2}$ & 8.22$_{\,\E+2}$ & 2.73$_{\,\E-1}$ \\ \addlinespace[1ex]
 \multirow{2}{*}{3} & Proposal & $-$8.50$_{\,\E+2}$ & 0\phantom{{}$_{\,\E+0}$} & 0\phantom{{}$_{\,\E+0}$} & 0\phantom{{}$_{\,\E+0}$} \\
 & \MakeUppercase{qic}$_{\text{\MakeUppercase{w}}}$ & $-$8.51$_{\,\E+3}$ & 6.13\phantom{{}$_{\,\E+0}$} & 3.35$_{\,\E+2}$ & 2.56$_{\,\E-1}$ \\ \addlinespace[2ex]
Block & & \multicolumn{1}{c}{black} & \multicolumn{1}{c}{hisp} & \multicolumn{1}{c}{married} & \multicolumn{1}{c}{nodegr} \\ \addlinespace[1ex]
 \multirow{2}{*}{1} & Proposal & 0\phantom{{}$_{\,\E+0}$} & 0\phantom{{}$_{\,\E+0}$} & 0\phantom{{}$_{\,\E+0}$} & 0\phantom{{}$_{\,\E+0}$} \\
 & \MakeUppercase{qic}$_{\text{\MakeUppercase{w}}}$ & 3.55$_{\,\E+2}$ & $-$8.19$_{\,\E+3}$ & $-$9.41$_{\,\E+2}$ & $-$1.30$_{\,\E+3}$ \\ \addlinespace[1ex]
 \multirow{2}{*}{2} & Proposal & $-$1.86$_{\,\E+4}$ & $-$2.34$_{\,\E+4}$ & 4.02$_{\,\E+3}$ & $-$1.26$_{\,\E+4}$ \\
 & \MakeUppercase{qic}$_{\text{\MakeUppercase{w}}}$ & $-$1.85$_{\,\E+4}$ & $-$2.42$_{\,\E+4}$ & 4.88$_{\,\E+3}$ & $-$1.34$_{\,\E+4}$ \\ \addlinespace[1ex]
 \multirow{2}{*}{3} & Proposal & 0\phantom{{}$_{\,\E+0}$} & $-$1.03$_{\,\E+4}$ & 0\phantom{{}$_{\,\E+0}$} & 0\phantom{{}$_{\,\E+0}$} \\
 & \MakeUppercase{qic}$_{\text{\MakeUppercase{w}}}$ & 3.84$_{\,\E+3}$ & $-$6.71$_{\,\E+3}$ & 2.92$_{\,\E+3}$ & $-$5.90$_{\,\E+2}$ \\
\end{tabular}
\end{center}
``$0$'' means that the covariate is not selected in the model selection.
\label{tab7}
\end{table}

\section{Discussion}
\label{sec7}
Our proposal, the \MakeUppercase{cbd} method, is an \MakeUppercase{sdid} approach incorporating covariate balancing. In order to ensure double robustness, i.e. to ensure that the \MakeUppercase{att} is consistently estimated when the model for the propensity score is correctly specified or the model for the changes in the outcomes over time is correctly supposed, we have shown that the second-order moments of the distribution for the covariates should be balanced between the treatment and control groups with the former as the target population.
Since covariate balancing usually balances the first-order moments, the value for the fact that the moment is second-order rather than first-order lies more in having been discovered due to its unexpectedness than in how to demonstrate its validity.
Moreover, our model selection criterion for the \MakeUppercase{cbd} method is an asymptotically unbiased estimator of risk based on the loss function used for the estimation. Specifically, we derived a model selection criterion using just a \MakeUppercase{did}-specific assumption called the conditional parallel trend assumption, i.e., without relying on the ignorable treatment assignment condition, which is often assumed in causal inference. The penalty term has been shown to be considerably different from twice the number of parameters that often appears in the \MakeUppercase{aic}-type model selection criteria. Since no model selection criteria existed for even the basic \MakeUppercase{sdid} approach, we developed a model selection criterion not only for the \MakeUppercase{cbd} method, but also for when the propensity scores are known or estimated by \MakeUppercase{mle}. Numerical experiments show that the proposed method estimates the \MakeUppercase{att} more robustly than the method that estimates propensity scores by using \MakeUppercase{mle} and that the proposed criterion clearly reduces the risk targeted by the \MakeUppercase{sdid} approach compared with the intuitive generalization of the existing information criterion to the \MakeUppercase{sdid} approach. In addition, analyses of real data confirm that there is a large difference between the results of the proposed method and those of the existing method.

Although this paper deals with a simple linear model of covariates as a working model for the outcome, its results can be easily extended to situations where a nonlinear model of covariates should be used, for example, when using the kernel method. In fact, \cite{hazlett20} has treated the function $\phi(x):\mathbb{R}^p \rightarrow \mathbb{R}^q$ in a reproducing kernel Hilbert space where $q$ may be larger than $p$ or $n$. They assumed that the outcome could be represented by a linear model of $\phi(x)$, which in our setting can be written as
\begin{align}
\E[y^{[1]}(1)-y^{[0]}(1)\mid x,d^{[1]}=1]=\phi(x)^{\T}\theta.
\label{hazlett}
\end{align}
Then, by using
\begin{align*}
\sum_{i=1}^n\bigg\{\frac{d^{[1]}_{i}\phi(x_i)}{e^{[1]}(x_{i};\alpha)} - \frac{d^{[0]}_{i}\phi(x_i)}{e^{[0]}(x_{i};\alpha)}\bigg\} = 0
\end{align*}
as the moment conditions for estimating the propensity scores, they propose the kernel balancing method, which can flexibly estimate causal effects even if the relationship between the covariate and the outcome is nonlinear. For making a doubly robust estimation using the kernel balancing in the setting of this paper, together with the assumption in \eqref{hazlett}, let us additionally suppose that the change in the outcomes follows
\begin{align*}
\E[\Delta^{[k]}\mid x,d^{[k]}=1]=\phi(x)^{\T}\beta^{[k]*}.
\end{align*}
Let $k(\cdot , \cdot) : \mathbb{R}^p \times \mathbb{R}^p \rightarrow \mathbb{R}$ be the kernel function corresponding to $\phi(x)$, and let $K$ be a $p\times p$ semi-definite symmetric matrix whose $(l,m)$-component is $k(x_l,x_m)$. If we estimate the parameter $\alpha$ that satisfies
\begin{align*}
& H^{[1]}(d,x;\alpha)=\E\bigg[e^{[1]}(x;\alpha)\bigg\{\frac{d^{[1]}}{e^{[1]}(x;\alpha)}-1\bigg\}K\bigg] = {\rm O},
\\
& H^{[0]}(d,x;\alpha)=\E\bigg[e^{[1]}(x;\alpha)\bigg\{\frac{d^{[0]}}{e^{[0]}(x;\alpha)}-1\bigg\}K\bigg] = {\rm O},
\end{align*}
and then estimate $\theta$ and derive a model selection criterion based on it, it would be a natural extension of the proposed method.

There are other valuable directions of study. For example, \cite{CalS21} and \cite{good21} proposed a \MakeUppercase{did} approach in a more flexible setting than the original one, where there are multiple time points and treatment can be initiated at any of them. Moreover, \cite{de2023sev} developed a \MakeUppercase{did} approach for when there are more than two groups, and \cite{rich2023gdid} proposed one to estimate causal effects under alternative assumptions to the parallel trend assumption. While the outcome in the basic \MakeUppercase{did} approach is continuous, \cite{tchetgen24} dealt with the case where the outcome follows a generalized linear model. In this way, the \MakeUppercase{did} approach has been theoretically developed in accordance with actual problems, so we believe that extending the proposed method to enable it to deal with those problems would thus be a potentially profitable goal for the near future. In econometrics, synthetic control (\citealt{AbaDH10}) has also become standard in causal inference, and a combination of the two has been developed as synthetic \MakeUppercase{did} (\citealt{ArkAHIW21}). Here, because the objective of the information criterion is what to synthesize and with what weights, the proposed method cannot necessarily be used; however, an investigation along the lines presented in this paper might prove fruitful. Furthermore, propensity score analysis with high-dimensional covariates is also an important topic (e.g., \citealt{BelCFH17}, \citealt{CheCDDHNR18}, \citealt{NinPI20}), wherein selecting covariates that truly have impact on causal effects and robustly estimating these effects is an essential task to ensure the validity of the analysis. Even when the covariates are high-dimensional, asymptotic properties of the estimator of causal effects tend to be obtained as well, so it is surely possible to extend our method. 

\appendix 

\section{Performance of model selection criterion for semiparametric difference-in-differences when propensity scores are known}
\label{secS1}

\subsection{Extension of \MakeUppercase{qic}$_{\text{\MakeUppercase{w}}}$}
\label{sec4_3}

As a comparator for the proposed model selection criterion, we will extend \MakeUppercase{qic}$_{\text{\MakeUppercase{w}}}$ (\citealt{PlaBCWS13}), which is an intuitive model selection criterion for weighted estimation using propensity scores, to the setting of this paper. First, we construct the goodness-of-fit term using the same idea as for \MakeUppercase{qic}$_{\text{\MakeUppercase{w}}}$. To do so, the log-likelihood function for full data is differentiated, a weighted potential outcome variable is substituted into the one whose expectation is the same as the outcome variable, the differentiation is reverted, and the estimator is substituted into the parameter. The number of parameters is used as the penalty term. Within this setting, the log-likelihood function for full data is
\begin{align*}
-\frac{1}{2\sigma^2}\sum_{i=1}^n[\{\mbox{$y^{[1]}_{i}(1)-y^{[0]}_{i}(1)$ when $d_i^{[1]}=1$}\}-x_i^{\T}\theta]^2.
\end{align*}
If we take $\rho_i\Delta_i$ as the weighted potential outcome variable, its conditional expectation can be shown to be the same as the conditional expectation of the outcome variable from \eqref{abaprop}. We will call the sum of the goodness-of-fit term and the penalty term multiplied by a constant $2\sigma^2$, that is,
\begin{align*}
\sum_{i=1}^{n}\{\rho(d_i,x_i)\Delta_i-x_{i}^{\T}\hat{\theta}\}^2 + 2\sigma^2p,
\end{align*}
the \MakeUppercase{qic}$_{\text{\MakeUppercase{w}}}$ for the \MakeUppercase{sdid} approach. In fact, if $\sigma^2(x)$ is a constant, then our proposed criterion is equivalent to the \MakeUppercase{qic}$_{\text{\MakeUppercase{w}}}$. The problem is to estimate the variance $\sigma^2$ of the difference, $y^{[1]}_{i}(1)-y^{[0]}_{i}(1)$ when $d_i^{[1]}=1$. This estimate cannot be obtained from the data using only the present assumptions. Since it follows from $y^{[0]}(0)=y^{[1]}(0)$ that $y^{[1]}_{i}(1)-y^{[0]}_{i}(1)=\Delta_{i}^{[1]}-\Delta_{i}^{[0]}$, and from Assumption \ref{hypo1}, the expectation of $\Delta_{i}^{[1]}-\Delta_{i}^{[0]}$ in $d^{[1]}_{i}=1$ is the same as the expectation of $\Delta_{i}^{[1]}$ in $d^{[0]}_{i}=0$ minus $\Delta_{i}^{[1]}-\Delta_{i}^{[0]}$ in $d^{[1]}_{i}=1$. However, in order to evaluate the variance, we will assume that $\Delta_{i}^{[1]}$ and $\Delta_{i}^{[0]}$ are independent conditioned on $d^{[1]}=1$ and $d^{[0]}=1$, from which we have
\begin{align*}
& \sigma^2 = {\rm var}[\Delta^{[1]}\mid d^{[1]}=1]+{\rm var}[\Delta^{[0]}\mid d^{[0]}=1]
\\
& = \E[\Delta^2\mid d^{[1]}=1] - \E[\Delta\mid d^{[1]}=1]^2 + \E[\Delta^2\mid d^{[0]}=1] - \E[\Delta\mid d^{[0]}=1]^2.
\end{align*}
Moreover, defining $n_1\equiv\sum_{i=1}^nd^{[1]}_{i}$ and $n_0\equiv\sum_{i=1}^nd^{[0]}_{i}$, we will use
\begin{align*}
\hat{\sigma}^2 = \frac{1}{n_1}\sum_{i=1}^nd^{[1]}_{i}\Delta_i^2-\bigg(\frac{1}{n_1}\sum_{i=1}^nd^{[1]}_{i}\Delta_i\bigg)^2+\frac{1}{n_0}\sum_{i=1}^nd^{[0]}_{i}\Delta_i^2-\bigg(\frac{1}{n_0}\sum_{i=1}^nd^{[0]}_{i}\Delta_i\bigg)^2
\end{align*}
in the present study.

\begin{remark}
Let us compare $\sigma^2(x)$ in the proposed criterion with $\sigma^2$ in the \MakeUppercase{qic}$_{\text{\MakeUppercase{w}}}$. In \eqref{std}, if all expectations do not actually depend on $x$, then
\begin{align*}
\sigma^2(x) = \ &  \E\bigg[\frac{\Delta^{[1]2}}{e^{[1]}}\ \bigg|\ d^{[1]}=1\bigg] - \E[\Delta^{[1]} \mid d^{[1]}=1]^2
\\
& + \E\bigg[\frac{\Delta^{[0]2}}{e^{[0]}}\ \bigg|\ d^{[0]}=1\bigg] - \E[\Delta^{[0]} \mid d^{[0]}=1]^2 + 2\E[\Delta^{[1]}\mid d^{[1]}=1]\E[\Delta^{[0]}\mid d^{[0]}=1].
\end{align*}
Since $\Delta^{[1]}$ and $\Delta^{[0]}$ include fluctuations over time, it is natural that $\Delta^{[1]}\Delta^{[0]}$ should be non-negative. In other words, basically $\sigma^2(x)>\sigma^2$, and the penalty term of the proposed criterion is larger than that of the \MakeUppercase{qic}$_{\text{\MakeUppercase{w}}}$.
\end{remark}

\subsection{Numerical experiments}
\label{sec4_4}
Here, the performance of the proposed criterion \eqref{IC_k} when the propensity scores are known is evaluated in simulation studies. The data generation settings include Case 1-1,
\begin{align*}
& x_1 \sim{\rm Uniform}(0, 2), \quad d^{[1]}\sim{\rm Bernoulli}({\rm logit}(-x_1)), \quad y^{[0]}(0)=y^{[1]}(0)\sim{\rm N}(0,1),
\\
& \epsilon^{[0]}, \epsilon^{[1]} \sim{\rm N}(0,1), \quad y^{[0]}(1)=y^{[0]}(0)+\epsilon^{[0]}, \quad y^{[1]}(1)=y^{[1]}(0)+1+\beta^* x_1+\epsilon^{[1]},
\end{align*}
and Case 1-2,
\begin{align*}
& x_1,x_2 \sim{\rm Uniform}(0, 2), \quad d^{[1]}\sim{\rm Bernoulli}({\rm logit}(-x_1+x_2)), \quad y^{[0]}(0)=y^{[1]}(0)\sim{\rm N}(0,1),
\\
& \epsilon^{[0]}, \epsilon^{[1]} \sim{\rm N}(0,1), \quad y^{[0]}(1)=y^{[0]}(0)+\epsilon^{[0]}, \quad y^{[1]}(1)=y^{[1]}(0)+1+\beta^* (x_1+x_2)+\epsilon^{[1]}.
\end{align*}
In both cases, $\beta^*$ represents the magnitude of contribution of the covariate $x_1$ to the causal effect.

Table \ref{tab1} shows whether the second term in \eqref{IC_k}, the penalty term of the proposed criterion, accurately approximates the bias in \eqref{pena} that we originally wanted to evaluate. For all values of $\beta^*$ and sample size $n$, our penalty terms are close to the true values of the bias, while the evaluated penalty terms in the \MakeUppercase{qic}$_{\text{\MakeUppercase{w}}}$ significantly underestimate the bias.

\begin{table}
\renewcommand{\baselinestretch}{1.5}\selectfont
\caption{Bias evaluation for the case where the propensity scores are known.}
\begin{center}
\begin{tabular}{rrrrrrrr}
& & \multicolumn{3}{c}{Case 1-1} & \multicolumn{3}{c}{Case 1-2} \\
 \multicolumn{1}{c}{$\beta^*$} & \multicolumn{1}{c}{$n$} & \multicolumn{1}{c}{True} & \multicolumn{1}{c}{Proposal} & \multicolumn{1}{c}{\MakeUppercase{qic}$_{\text{\MakeUppercase{w}}}$} & \multicolumn{1}{c}{True} & \multicolumn{1}{c}{Proposal} & \multicolumn{1}{c}{\MakeUppercase{qic}$_{\text{\MakeUppercase{w}}}$} \\ \addlinespace[1ex]
 & 200 & ~37.54 & ~37.29 & 2.23 & ~39.20 & ~39.43 & ~5.96 \\
 0.1       & 400 & ~38.56 & ~37.71 & 2.24 & ~39.75 & ~39.74 & ~5.98 \\
       & 600 & ~37.89 & ~37.89 & 2.27 & ~39.79 & ~39.68 & ~6.00 \\ \addlinespace[1ex]
 & 200 & ~57.34 & ~56.28 & 2.31 & ~59.77 & ~58.69 & ~6.39 \\
 0.5       & 400 & ~59.36 & ~56.28 & 2.33 & ~59.36 & ~59.52 & ~6.48 \\
       & 600 & ~58.56 & ~57.35 & 2.34 & ~60.20 & ~59.37 & ~6.48 \\ \addlinespace[1ex]
 & 200 & ~92.48 & ~91.33 & 2.56 & 103.24 & ~98.85 & ~7.91 \\
 1.0       & 400 & ~98.17 & ~91.43 & 2.58 & 108.09 & ~99.46 & ~7.99 \\
       & 600 & ~87.03 & ~92.62 & 2.59 & ~98.33 & ~99.57 & ~7.99 \\ \addlinespace[1ex]
 & 200 & 359.64 & 358.56 & 5.17 & 439.35 & 429.91 & 23.65 \\
 3.0       & 400 & 345.73 & 348.96 & 5.20 & 441.47 & 430.88 & 24.05 \\
       & 600 & 367.73 & 365.75 & 5.22 & 435.32 & 430.09 & 23.87 \\ 
\end{tabular}
\end{center}
True, the true value of the bias evaluated by the Monte Carlo method; Proposal, the evaluated penalty term of the proposed criterion; \MakeUppercase{qic}$_{\text{\MakeUppercase{w}}}$, the evaluated penalty term of the \MakeUppercase{qic}$_{\text{\MakeUppercase{w}}}$.
\label{tab1}
\end{table}

Next, let us examine the performance of the model selection using the proposed criterion, with data generated from
\begin{align*}
& x_1,\ldots,x_l\sim{\rm Uniform}(0, 2), \quad x=(1,x_1,\ldots,x_{l})^{\T}, \quad y^{[0]}(0)=y^{[1]}(0)\sim{\rm N}(0,1),
\\
& \epsilon^{[0]},\epsilon^{[1]} \sim{\rm N}(0,1), \quad y^{[0]}(1)=y^{[0]}(0)+\epsilon^{[0]}, \quad y^{[1]}(1)=y^{[1]}(0)+x^{\T}\beta^{[1]*}+\epsilon^{[1]},
\end{align*}
in Case 2-1,
\begin{align*}
l=4, \quad \beta^{[1]*} = (1,\beta^*,0_{3}^{\T})^{\T}, \quad d^{[1]}\sim{\rm Bernoulli}({\rm logit}(-x_1)),
\end{align*}
Case 2-2,
\begin{align*}
l=4, \quad \beta^{[1]*} = (1,1_2^{\T}\beta^*,0_{2}^{\T})^{\T}, \quad d^{[1]}\sim{\rm Bernoulli}({\rm logit}(-x_1+x_2)),
\end{align*}
and Case 2-3, 
\begin{align*}
l=6, \quad \beta^{[1]*} = (1,1_2^{\T}\beta^*,0_{4}^{\T})^{\T}, \quad d^{[1]}\sim{\rm Bernoulli}({\rm logit}(-x_1+x_2)).
\end{align*}
Table \ref{tab2} compares the proposed criterion and the \MakeUppercase{qic}$_{\text{\MakeUppercase{w}}}$ when using the forward selection method. As the main index to measure the goodness of the criteria, we use an empirical estimate of the risk, specifically, the average of 3,000 values of
\begin{align*}
\sum_{i=1}^{n}e^{[1]}(x_i;\alpha)(x_{i}^{\T}\theta^{*}-x_{i}^{\T}\hat{\theta})^2.
\end{align*}
In Case 2-2, where there are relatively few unnecessary covariates among the candidates, that is, where the negative effect of false positives is less likely to appear, the \MakeUppercase{qic}$_{\text{\MakeUppercase{w}}}$ slightly outperforms the proposed criterion for $\beta^*=0.5$ and $\beta^*=1.0$. However, the proposed criterion outperforms the \MakeUppercase{qic}$_{\text{\MakeUppercase{w}}}$ for the other values of beta in Case 2-2, and it outperforms the \MakeUppercase{qic}$_{\text{\MakeUppercase{w}}}$ for all values in Case 2-1 and Case 2-3, i.e., settings with many covariates are unrelated to the true structure. Here, since the \MakeUppercase{qic}$_{\text{\MakeUppercase{w}}}$ underestimates the penalty term, the false positive must be large and too many covariates are selected in all settings.

\begin{table}
\renewcommand{\baselinestretch}{1.5}\selectfont
\caption{Comparison of the proposed criterion and the \MakeUppercase{qic}$_{\text{\MakeUppercase{w}}}$ through model selection for case in which the propensity scores are known.}
\begin{center}
\begin{tabular}{rrcrrrrrrrrr}
 & & & \multicolumn{3}{c}{Case 2-1} & \multicolumn{3}{c}{Case 2-2} & \multicolumn{3}{c}{Case 2-3} \\
\multicolumn{1}{c}{$\beta^*$} & \multicolumn{1}{c}{$n$} & & \multicolumn{1}{c}{Risk} & \multicolumn{1}{c}{\MakeUppercase{tp}} & \multicolumn{1}{c}{\MakeUppercase{fp}} & \multicolumn{1}{c}{Risk} & \multicolumn{1}{c}{\MakeUppercase{tp}} & \multicolumn{1}{c}{\MakeUppercase{fp}} & \multicolumn{1}{c}{Risk} & \multicolumn{1}{c}{\MakeUppercase{tp}} & \multicolumn{1}{c}{\MakeUppercase{fp}} \\ \addlinespace[1ex]
 & \multirow{2}{*}{200} & Proposal & 11.05 & 0.50 & 1.29 & ~14.07 & 0.61 & 0.62 & ~19.01 & 0.63 & 1.15 \\
                               & & \MakeUppercase{qic}$_{\text{\MakeUppercase{w}}}$ & 11.29 & 0.53 & 1.47 & ~15.94 & 0.90 & 0.95 & ~21.74 & 0.90 & 1.73 \\ \addlinespace[1ex]
 \multirow{2}{*}{0.1} & \multirow{2}{*}{400} & Proposal & 11.04 & 0.47 & 1.30 & ~13.99 & 0.68 & 0.52 & ~19.72 & 0.69 & 1.08 \\
                               & & \MakeUppercase{qic}$_{\text{\MakeUppercase{w}}}$ & 11.27 & 0.51 & 1.49 & ~15.23 & 0.98 & 0.86 & ~22.45 & 0.98 & 1.77 \\ \addlinespace[1ex]
                    & \multirow{2}{*}{600} & Proposal & 10.86 & 0.48 & 1.25 & ~14.40 & 0.75 & 0.51 & ~19.51 & 0.72 & 1.14 \\
                               & & \MakeUppercase{qic}$_{\text{\MakeUppercase{w}}}$ & 11.00 & 0.51 & 1.44 & ~15.88 & 1.07 & 0.83 & ~22.62 & 1.02 & 1.81 \\
\addlinespace[1ex]
 & \multirow{2}{*}{200} & Proposal & 14.68 & 0.70 & 1.35 & ~23.20 & 1.23 & 0.61 & ~30.01 & 1.20 & 1.21 \\
                               & & \MakeUppercase{qic}$_{\text{\MakeUppercase{w}}}$ & 14.94 & 0.74 & 1.61 & ~23.33 & 1.50 & 1.00 & ~31.61 & 1.48 & 1.96 \\ \addlinespace[1ex]
\multirow{2}{*}{0.5} & \multirow{2}{*}{400} & Proposal & 14.42 & 0.83 & 1.29 & ~22.95 & 1.60 & 0.60 & ~31.82 & 1.55 & 1.14 \\
                               & & \MakeUppercase{qic}$_{\text{\MakeUppercase{w}}}$ & 14.71 & 0.86 & 1.58 & ~22.48 & 1.78 & 1.01 & ~32.67 & 1.75 & 1.93 \\ \addlinespace[1ex]
                    & \multirow{2}{*}{600} & Proposal & 14.67 & 0.89 & 1.31 & ~23.24 & 1.74 & 0.57 & ~30.25 & 1.75 & 1.20 \\
                               & & \MakeUppercase{qic}$_{\text{\MakeUppercase{w}}}$ & 14.93 & 0.91 & 1.61 & ~23.00 & 1.88 & 0.98 & ~31.96 & 1.88 & 1.98 \\
\addlinespace[1ex]
 & \multirow{2}{*}{200} & Proposal & 21.69 & 0.89 & 1.39 & ~37.19 & 1.69 & 0.61 & ~48.24 & 1.69 & 1.26 \\
                               & & \MakeUppercase{qic}$_{\text{\MakeUppercase{w}}}$ & 22.00 & 0.91 & 1.67 & ~36.59 & 1.85 & 1.04 & ~49.53 & 1.84 & 2.18 \\ \addlinespace[1ex]
\multirow{2}{*}{1.0} & \multirow{2}{*}{400} & Proposal & 21.09 & 0.97 & 1.40 & ~34.65 & 1.95 & 0.63 & ~46.07 & 1.94 & 1.20 \\
                               & & \MakeUppercase{qic}$_{\text{\MakeUppercase{w}}}$ & 21.61 & 0.98 & 1.76 & ~35.71 & 1.99 & 1.08 & ~49.45 & 1.98 & 2.12 \\ \addlinespace[1ex]
                    & \multirow{2}{*}{600} & Proposal & 20.74 & 1.00 & 1.35 & ~32.16 & 1.99 & 0.60 & ~43.66 & 2.00 & 1.24 \\
                               & & \MakeUppercase{qic}$_{\text{\MakeUppercase{w}}}$ & 21.30 & 1.00 & 1.72 & ~34.38 & 1.99 & 1.07 & ~48.02 & 2.00 & 2.19 \\
\addlinespace[1ex]
 & \multirow{2}{*}{200} & Proposal & 67.79 & 0.99 & 1.54 & 136.61 & 1.96 & 0.67 & 180.41 & 1.96 & 1.33 \\
                               & & \MakeUppercase{qic}$_{\text{\MakeUppercase{w}}}$ & 69.22 & 1.00 & 1.93 & 141.09 & 1.99 & 1.19 & 194.02 & 1.98 & 2.34 \\ \addlinespace[1ex]
\multirow{2}{*}{3.0} & \multirow{2}{*}{400} & Proposal & 66.20 & 1.00 & 1.57 & 133.32 & 2.00 & 1.20 & 175.10 & 2.00 & 1.25 \\
                               & & \MakeUppercase{qic}$_{\text{\MakeUppercase{w}}}$ & 67.83 & 1.00 & 1.97 & 141.74 & 2.00 & 1.20 & 191.79 & 2.00 & 2.38 \\ \addlinespace[1ex]
                    & \multirow{2}{*}{600} & Proposal & 65.76 & 1.00 & 1.48 & 126.36 & 2.00 & 0.62 & 173.78 & 2.00 & 1.29 \\
                               & & \MakeUppercase{qic}$_{\text{\MakeUppercase{w}}}$ & 67.51 & 1.00 & 1.91 & 135.24 & 2.00 & 1.17 & 191.42 & 2.00 & 2.36 \\
\end{tabular}
\end{center}
Risk, the empirical estimate of the risk of the selected model; \MakeUppercase{tp}, the average number of covariates that are true positive; \MakeUppercase{fp}, the average number of covariates that are false positive.
\label{tab2}
\end{table}

\section{Model selection criterion for semiparametric difference-in-differences via maximum likelihood estimation}\label{secA}

\subsection{Derivation of model selection criterion}\label{secA_1}

For the parameter $\alpha$ in the propensity scores, the maximum likelihood estimator (\MakeUppercase{mle}) $\hat{\alpha}^{\rm \MakeUppercase{ml}}$ is obtained by maximizing $\sum_{i=1}^{n}\{d_{i}^{[0]}\log e^{[0]}( x_i;\alpha)+d_{i}^{[1]}\log e^{[1]}( x_i;\alpha)\}$.
Denoting the Fisher information matrix based on this likelihood as $I(\alpha)$, the error of $\hat{\alpha}^{\rm \MakeUppercase{ml}}$ multiplied by $\sqrt{\vphantom{x}}{n}$ can be written as
\begin{align}
&\sqrt{\vphantom{x}}{n}(\hat{\alpha}^{\rm \MakeUppercase{ml}}-\alpha^{*}) \notag \\
&=\frac{1}{\sqrt{\vphantom{x}}{n}}I(\alpha^{*})^{-1}\sum_{i=1}^{n}\biggl\{d_{i}^{[0]}\dfrac{\partial}{\partial \alpha}\log e^{[0]}( x_i;\alpha^{*})+d_{i}^{[1]}\dfrac{\partial}{\partial \alpha}\log e^{[1]}( x_i;\alpha^{*})\biggr\} + \oP(1).
\label{mlediff}
\end{align}
Let us derive the model selection criterion when the propensity scores are estimated by the \MakeUppercase{mle}, using the same approach as the model selection criterion for the covariate balancing in the difference-in-differences (\MakeUppercase{cbd}) method.
In its derivation, we need to evaluate the third term of the risk in \eqref{pena}, 
\begin{align*}
\sum_{i=1}^n 2\E[e^{[1]}( x_i;\hat{\alpha}^{\rm \MakeUppercase{ml}})\{\rho( d_i, x_i;\hat{\alpha}^{\rm \MakeUppercase{ml}})\Delta_i-\E[y^{[1]}_{i}-y^{[0]}_{i} \mid  x_i, d^{[1]}_{i}=1]\} x_{i}^{\T}(\hat{\theta}^{\rm \MakeUppercase{ml}}-\theta^*)].
\end{align*}
Substituting $ x_i^{\T}\theta^*$ into $\E[y^{[1]}_{i}-y^{[0]}_{i} \mid  x_i, d^{[1]}_{i}=1]$ first, and then substituting the term corresponding to the right-hand side of \eqref{expand} into $\hat{\theta}^{\rm \MakeUppercase{ml}}-\theta^*$, yields
\begin{align*}
& 2\E\bigg[\frac{1}{\sqrt{\vphantom{x}}{n}}\sum_{i=1}^n e^{[1]}( x_{i};\hat{\alpha}^{\rm \MakeUppercase{ml}})\{\rho( d_i, x_i;\hat{\alpha}^{\rm \MakeUppercase{ml}})\Delta_i- x_i^{\T}\theta^*\} x_i^{\T} \bigg\{\frac{1}{n}\sum_{j=1}^ne^{[1]}( x_j;\hat{\alpha}^{\rm \MakeUppercase{ml}}) x_j x_j^{\T}\bigg\}^{-1} 
\\
& \ \phantom{\E\bigg[} \times\frac{1}{\sqrt{\vphantom{x}}{n}}\sum_{j=1}^n e^{[1]}( x_j;\hat{\alpha}^{\rm \MakeUppercase{ml}})\{\rho( d_j, x_j;\hat{\alpha}^{\rm \MakeUppercase{ml}})\Delta_j- x_j^{\T}\theta^*\} x_j+\oP(1)\bigg].
\end{align*}
Since $n^{-1}\sum_{j=1}^ne^{[1]}( x_j;\hat{\alpha}^{\rm \MakeUppercase{ml}}) x_j x_j^{\T}$ converges in probability to
\begin{align*}
L(\alpha^{*})\equiv\E[e^{[1]}( x;\alpha^{*}) x x^{\T}],
\end{align*}
by letting $v$ be the Gaussian random vector that is the weak limit of
\begin{align}
\frac{1}{\sqrt{\vphantom{x}}{n}}\sum_{i=1}^ne^{[1]}( x_i;\hat{\alpha}^{\rm \MakeUppercase{ml}})\{\rho(d_i, x_i;\hat{\alpha}^{\rm \MakeUppercase{ml}})\Delta_i- x_i^{\T}\theta^*\} x_i,
\label{weaks}
\end{align}
we can express the weak limit of \eqref{pena} with expectation removed as
\begin{align*}
c^{\rm limit} = 2v^{\T}L(\alpha^{*})^{-1}v = 2{\rm tr}\{L(\alpha^{*})^{-1}vv^{\T}\}.
\end{align*}
When $\hat{\theta}^{\rm \MakeUppercase{ml}}$ is a consistent estimator of $\theta^*$, it holds that ${\rm E}[vv^{\T}]={\rm var}[v]$ since the expectation of $v$ becomes $0$.
Then, the asymptotic evaluation of \eqref{pena}, which is to be used in the information criterion, is written as
\begin{align*}
\E[c^{\rm limit}] =  2{\rm tr}\{L(\alpha^{*})^{-1}{\rm var}[v]\}.
\end{align*}

Now let us derive a specific expression for ${\rm var}[v]$.
Since the error of the parameter $\hat{\alpha}^{\rm \MakeUppercase{ml}}$ in the propensity scores is written as \eqref{mlediff}, we can expand \eqref{weaks} as
\begin{align*}
& \frac{1}{\sqrt{\vphantom{x}}{n}}\sum_{i=1}^ne^{[1]}( x_i;\hat{\alpha}^{\rm \MakeUppercase{ml}})\{\rho(d_i, x_i;\hat{\alpha}^{\rm \MakeUppercase{ml}})\Delta_i- x_i^{\T}\theta^*\} x_i
\\
& = \frac{1}{\sqrt{\vphantom{x}}{n}}\sum_{i=1}^{n} x_i\bigg\{e^{[1]}( x_i;\alpha^{*})\{\rho(d_i, x_i;\alpha^{*})\Delta_i- x_i^{\T}\theta^*\}
\\
& \ \phantom{= \frac{1}{\sqrt{\vphantom{x}}{n}}\sum_{i=1}^{n} x_i} + \frac{\partial}{\partial\alpha^{\T}}\bigg(e^{[1]}( x_i;\alpha^{*})\bigg[\frac{d^{[1]}_{i}-e^{[1]}( x_i;\alpha^{*})}{e^{[1]}( x_i;\alpha^{*})\{1-e^{[1]}( x_i;\alpha^{*})\}}\Delta_i- x_i^{\T}\theta^*\bigg]\bigg)(\hat{\alpha}^{\rm \MakeUppercase{ml}}-\alpha^{*})\bigg\}
\\
& \ \phantom{= \frac{1}{\sqrt{\vphantom{x}}{n}}\sum_{i=1}^{n} x_i} \times\{1+\oP(1)\}
\\
& = \frac{1}{\sqrt{\vphantom{x}}{n}}\sum_{i=1}^{n}\Big[e^{[1]}( x_i;\alpha^{*})\{\rho(d_i, x_i;\alpha^{*})\Delta_i- x_i^{\T}\theta^*\} x_{i} \\
& \ \phantom{= \frac{1}{\sqrt{\vphantom{x}}{n}}} - M(\alpha^{*},\theta^{*})I(\alpha^{*})^{-1}\Big\{d_i^{[0]}\dfrac{\partial}{\partial \alpha}\log e^{[0]}( x_i;\alpha^{*})+d_i^{[1]}\frac{\partial}{\partial\alpha}\log e^{[1]}( x_i;\alpha^{*})\Big\}\Big]\{1+\oP(1)\}.
\end{align*}
Then, by defining
\begin{align}
V^{\rm \MakeUppercase{ml}}(\alpha^{*},\theta^*) & \equiv {\rm var}\Big[e^{[1]}( x;\alpha^{*})\{\rho(d, x;\alpha^{*})\Delta- x^{\T}\theta^*\} x
\notag \\
& \ \phantom{\equiv {\rm var}[} - M(\alpha^{*},\theta^{*})I(\alpha^{*})^{-1}\Big\{d^{[0]}\frac{\partial}{\partial\alpha}\log e^{[0]}( x;\alpha^{*})+d^{[1]}\frac{\partial}{\partial\alpha}\log e^{[1]}( x;\alpha^{*})\Big\}\Big],
\label{defV}
\end{align}
it can be shown that ${\rm var}[v]=V^{\rm \MakeUppercase{ml}}(\alpha^{*},\theta^*)$.
Thus, the following theorem is obtained.

\begin{theorem}
When the propensity scores are estimated using the \MakeUppercase{mle}, the third term in \eqref{risk} is evaluated as
\begin{align*}
\E[c^{\rm limit}] =  2{\rm tr}\{L(\alpha^{*})^{-1}V^{\rm \MakeUppercase{ml}}(\alpha^{*},\theta^*)\},
\end{align*}
where $L(\cdot)$ and $V^{\rm \MakeUppercase{ml}}(\cdot,\cdot)$ are defined in \eqref{defL} and \eqref{defV}, respectively. 
\end{theorem}

Since the penalty term in this theorem depends on the true values of the parameters, we can use empirical estimates in practice.
Specifically, we will replace $L(\alpha^*)$ by $L_n(\hat{\alpha}^{\rm \MakeUppercase{ml}}) \equiv n^{-1}\sum_{i=1}^ne^{[1]}( x_i;\allowbreak\hat{\alpha}^{\rm \MakeUppercase{ml}}) x_i x_i^{\T}$ and $V^{\rm \MakeUppercase{ml}}(\alpha^*,\theta^*)$ by 
\begin{align*}
& V_n^{\rm \MakeUppercase{ml}}(\hat{\alpha}^{\rm \MakeUppercase{ml}},\hat{\theta}^{\rm \MakeUppercase{ml}})
\\
& \equiv \frac{1}{n}\sum_{i=1}^n\Big[e^{[1]}( x_i;\hat{\alpha}^{\rm \MakeUppercase{ml}})\{\rho(d_i, x_i;\hat{\alpha}^{\rm \MakeUppercase{ml}})\Delta_i- x_i^{\T}\hat{\theta}^{\rm \MakeUppercase{ml}}\} x_i
\notag \\
& \ \phantom{\equiv \frac{1}{n}\sum_{i=1}^n} - M_n(\hat{\alpha}^{\rm \MakeUppercase{ml}},\hat{\theta}^{\rm \MakeUppercase{ml}})I(\hat{\alpha}^{\rm \MakeUppercase{ml}})^{-1}\Big\{d_i^{[0]}\frac{\partial}{\partial\alpha}\log e^{[0]}( x_i;\hat{\alpha}^{\rm \MakeUppercase{ml}})+d_i^{[1]}\frac{\partial}{\partial\alpha}\log e^{[1]}( x_i;\hat{\alpha}^{\rm \MakeUppercase{ml}})\Big\}\Big],
\end{align*}
where $M_n(\cdot,\cdot)$ is defined in \eqref{defMn}.
As a result, for when the propensity scores are estimated with the \MakeUppercase{mle}, we propose
\begin{align}
\sum_{i=1}^{n}e^{[1]}( x_{i};\hat{\alpha}^{\rm \MakeUppercase{ml}})\{\rho(d_i, x_{i};\hat{\alpha}^{\rm \MakeUppercase{ml}})\Delta_i- x_{i}^{\T}\hat{\theta}^{\rm \MakeUppercase{ml}}\}^2 + 2{\rm tr}\{L_n(\hat{\alpha}^{\rm \MakeUppercase{ml}})^{-1}V_n^{\rm \MakeUppercase{ml}}(\hat{\alpha}^{\rm \MakeUppercase{ml}},\hat{\theta}^{\rm \MakeUppercase{ml}})\}
\label{ICMLE}
\end{align}
as the model selection criterion, which is an asymptotically unbiased estimator of the risk in \eqref{risk}.

\subsection{Numerical experiments}\label{secA_2}

The simulation studies presented in this section examine the accuracy of the approximation for the penalty term and the performance of the model selection criterion in \eqref{ICMLE}.
The simulation setting is the same as in Section \ref{sec4_4}, and the number of repetitions is 3,000.
Table \ref{tabs1} lists results for data generated as in Case 1-1 and Case 1-2 in Section \ref{sec4_4} to check whether the second term in \eqref{ICMLE}, the penalty term of the proposed criterion, accurately approximates the bias in \eqref{pena}, which we originally wanted to evaluate, for various values of the parameter $\beta^*$ and sample size $n$.
In all cases, the evaluated penalty terms of the proposed criterion are much closer to the bias than the evaluated penalty terms of the \MakeUppercase{qic}$_{\text{\MakeUppercase{w}}}$, and its approximation accuracy is high.
In contrast, the penalty term of the \MakeUppercase{qic}$_{\text{\MakeUppercase{w}}}$ significantly underestimates the bias in all cases.

\begin{table}[!t]
\renewcommand{\baselinestretch}{1.5}\selectfont
\caption{Bias evaluation for case in which the propensity scores are estimated using the \MakeUppercase{mle}.}
\begin{center}
\begin{tabular}{rrrrrrrr}
 & & \multicolumn{3}{c}{Case 1-1} & \multicolumn{3}{c}{Case 1-2} \\
 \multicolumn{1}{c}{$\beta^*$} & \multicolumn{1}{c}{$n$} & \multicolumn{1}{c}{True} & \multicolumn{1}{c}{Proposal} & \multicolumn{1}{c}{\MakeUppercase{qic}$_{\text{\MakeUppercase{w}}}$} & \multicolumn{1}{c}{True} & \multicolumn{1}{c}{Proposal} & \multicolumn{1}{c}{\MakeUppercase{qic}$_{\text{\MakeUppercase{w}}}$} \\
\addlinespace[1ex]
       & 200 & 7.48 & 7.33 & 2.23 & 17.91 & 17.25 & 5.96 \\
 0.1   & 400 & 7.57 & 7.29 & 2.24 & 17.79 & 17.48 & 5.98 \\
       & 600 & 7.38 & 7.33 & 2.27 & 17.14 & 17.40 & 6.00 \\
\addlinespace[1ex]
       & 200 & 8.26 & 8.09 & 2.31 & 19.47 & 18.86 & 6.39 \\
 0.5   & 400 & 7.88 & 8.01 & 2.33 & 19.60 & 19.40 & 6.48 \\
       & 600 & 7.75 & 8.03 & 2.34 & 18.95 & 19.25 & 6.48 \\
\addlinespace[1ex]
       & 200 & 9.78 & 9.43 & 2.58 & 24.67 & 22.89 & 7.91 \\
 1.0   & 400 & 9.17 & 9.28 & 2.58 & 23.01 & 23.19 & 7.99 \\
       & 600 & 9.21 & 9.27 & 2.59 & 23.41 & 23.20 & 7.99 \\
\addlinespace[1ex]
       & 200 & 19.12 & 19.31 & 5.17 & 54.75 & 54.89 & 23.66 \\
 3.0   & 400 & 18.46 & 18.63 & 5.20 & 53.13 & 54.85 & 24.04 \\
       & 600 & 17.79 & 18.53 & 5.22 & 53.79 & 54.66 & 23.87 \\
\end{tabular}
\end{center}
True, the true value of the bias evaluated by the Monte Carlo method; Proposal, the penalty term of the proposed criterion; \MakeUppercase{qic}$_{\text{\MakeUppercase{w}}}$, the penalty term of the \MakeUppercase{qic}$_{\text{\MakeUppercase{w}}}$.
\label{tabs1}
\end{table}

Table \ref{tabs2} lists results for data generated as in Case 2-1, Case 2-2, and Case 2-3 in Section \ref{sec4_4} and comparing the proposed criterion and the \MakeUppercase{qic}$_{\text{\MakeUppercase{w}}}$ through the forward selection method.
The main index of goodness for the model selection criteria is the empirical estimate of the risk.
Similarly to Section \ref{sec5_2}, where propensity scores are estimated using the \MakeUppercase{cbd} method, the proposed criterion always outperforms the \MakeUppercase{qic}$_{\text{\MakeUppercase{w}}}$ in all cases.
The \MakeUppercase{qic}$_{\text{\MakeUppercase{w}}}$ underestimates the penalty term; thus, too many covariates are selected.

\begin{table}
\renewcommand{\baselinestretch}{1.5}\selectfont
\caption{Comparison of the proposed criterion and the \MakeUppercase{qic}$_{\text{\MakeUppercase{w}}}$ through model selection for case in which the propensity scores are estimated using the \MakeUppercase{mle}.}
\begin{center}
\begin{tabular}{rrrrrrrrrrrr}
 & & & \multicolumn{3}{c}{Case 2-1} & \multicolumn{3}{c}{Case 2-2} & \multicolumn{3}{c}{Case 2-3} \\
 \multicolumn{1}{c}{$\beta^*$} & \multicolumn{1}{c}{$n$} & & \multicolumn{1}{c}{Risk} & \multicolumn{1}{c}{\MakeUppercase{tp}} & \multicolumn{1}{c}{\MakeUppercase{fp}} & \multicolumn{1}{c}{Risk} & \multicolumn{1}{c}{\MakeUppercase{tp}} & \multicolumn{1}{c}{\MakeUppercase{fp}} & \multicolumn{1}{c}{Risk} & \multicolumn{1}{c}{\MakeUppercase{tp}} & \multicolumn{1}{c}{\MakeUppercase{fp}} \\
\addlinespace[1ex]
 & \multirow{2}{*}{200} & \multicolumn{1}{c}{Proposal} & 7.67 & 0.41 & 1.26 & 9.38 & 0.62 & 0.61 & 11.61 & 0.48 & 1.01 \\
 & & \multicolumn{1}{c}{\MakeUppercase{qic}$_{\text{\MakeUppercase{w}}}$} & 10.35 & 0.47 & 1.55 & 14.68 & 0.87 & 0.90 & 20.09 & 0.86 & 1.76 \\
\addlinespace[1ex]
 \multirow{2}{*}{0.1} & \multirow{2}{*}{400} & \multicolumn{1}{c}{Proposal} & 7.70 & 0.42 & 1.28 & 9.44 & 0.58 & 0.55 & 12.39 & 0.50 & 0.97 \\
 & & \multicolumn{1}{c}{\MakeUppercase{qic}$_{\text{\MakeUppercase{w}}}$} & 10.44 & 0.47 & 1.46 & 14.92 & 0.92 & 0.90 & 21.04 & 0.95 & 1.74 \\
\addlinespace[1ex]
 & \multirow{2}{*}{600} & \multicolumn{1}{c}{Proposal} & 7.45 & 0.42 & 1.24 & 9.66 & 0.63 & 0.54 & 12.16 & 0.51 & 0.89 \\
 & & \multicolumn{1}{c}{\MakeUppercase{qic}$_{\text{\MakeUppercase{w}}}$} & 9.89 & 0.49 & 1.43 & 14.93 & 0.98 & 0.91 & 21.18 & 0.97 & 1.79 \\
\addlinespace[1ex]
 & \multirow{2}{*}{200} & \multicolumn{1}{c}{Proposal} & 10.68 & 0.53 & 1.2 & 18.96 & 1.28 & 0.84 & 23.15 & 0.75 & 0.92 \\
 & & \multicolumn{1}{c}{\MakeUppercase{qic}$_{\text{\MakeUppercase{w}}}$} & 13.13 & 0.77 & 1.64 & 19.50 & 1.50 & 1.02 & 27.59 & 1.50 & 1.97 \\
\addlinespace[1ex]
 \multirow{2}{*}{0.5} & \multirow{2}{*}{400} & \multicolumn{1}{c}{Proposal} & 11.55 & 0.67 & 1.24 & 20.57 & 1.58 & 0.86 & 26.47 & 1.45 & 1.40 \\
 & & \multicolumn{1}{c}{\MakeUppercase{qic}$_{\text{\MakeUppercase{w}}}$} & 13.17 & 0.88 & 1.61 & 19.44 & 1.82 & 1.00 & 27.82 & 1.81 & 1.94 \\
\addlinespace[1ex]
 & \multirow{2}{*}{600} & \multicolumn{1}{c}{Proposal} & 11.20 & 0.77 & 1.21 & 20.33 & 1.75 & 0.88 & 26.66 & 1.64 & 1.41 \\
 & & \multicolumn{1}{c}{\MakeUppercase{qic}$_{\text{\MakeUppercase{w}}}$} & 12.27 & 0.97 & 1.60 & 19.34 & 1.93 & 1.02 & 27.50 & 1.93 & 2.04 \\
\addlinespace[1ex]
 & \multirow{2}{*}{200} & \multicolumn{1}{c}{Proposal} & 15.48 & 0.75 & 1.19 & 26.22 & 1.63 & 0.81 & 37.08 & 1.54 & 1.30 \\
 & & \multicolumn{1}{c}{\MakeUppercase{qic}$_{\text{\MakeUppercase{w}}}$} & 16.98 & 0.95 & 1.77 & 25.71 & 1.93 & 1.08 & 39.01 & 1.94 & 2.17 \\
\addlinespace[1ex]
 \multirow{2}{*}{1.0} & \multirow{2}{*}{400} & \multicolumn{1}{c}{Proposal} & 14.89 & 0.90 & 1.25 & 24.21 & 1.92 & 0.82 & 32.27 & 1.86 & 1.28 \\
 & & \multicolumn{1}{c}{\MakeUppercase{qic}$_{\text{\MakeUppercase{w}}}$} & 16.66 & 1.00 & 1.76 & 25.29 & 2.00 & 1.09 & 38.57 & 2.00 & 2.15 \\
\addlinespace[1ex]
 & \multirow{2}{*}{600} & \multicolumn{1}{c}{Proposal} & 12.35 & 0.98 & 1.23 & 21.63 & 1.99 & 0.85 & 28.22 & 1.96 & 1.28 \\
 & & \multicolumn{1}{c}{\MakeUppercase{qic}$_{\text{\MakeUppercase{w}}}$} & 15.88 & 1.00 & 1.76 & 25.46 & 2.00 & 1.09 & 37.98 & 2.00 & 2.17 \\
\addlinespace[1ex]
 & \multirow{2}{*}{200} & \multicolumn{1}{c}{Proposal} & 43.16 & 0.93 & 1.29 & 75.75 & 1.92 & 0.82 & 97.94 & 1.92 & 1.35 \\
 & & \multicolumn{1}{c}{\MakeUppercase{qic}$_{\text{\MakeUppercase{w}}}$} & 45.70 & 1.00 & 1.94 & 76.60 & 2.00 & 1.16 & 134.42 & 2.00 & 2.37 \\
\addlinespace[1ex]
 \multirow{2}{*}{3.0} & \multirow{2}{*}{400} & \multicolumn{1}{c}{Proposal} & 39.81 & 0.98 & 1.41 & 61.92 & 1.99 & 0.83 & 79.53 & 2.00 & 1.34 \\
 & & \multicolumn{1}{c}{\MakeUppercase{qic}$_{\text{\MakeUppercase{w}}}$} & 45.46 & 1.00 & 0.51 & 79.06 & 2.00 & 1.16 & 132.36 & 2.00 & 2.36 \\
\addlinespace[1ex]
 & \multirow{2}{*}{600} & \multicolumn{1}{c}{Proposal} & 33.37 & 1.00 & 1.37 & 56.02 & 2.00 & 0.82 & 78.31 & 2.00 & 1.27 \\
 & & \multicolumn{1}{c}{\MakeUppercase{qic}$_{\text{\MakeUppercase{w}}}$} & 43.90 & 1.00 & 1.95 & 78.45 & 2.00 & 1.18 & 129.13 & 2.00 & 2.34 \\
\end{tabular}
\end{center}
Risk, the empirical estimate of the risk of the selected model; \MakeUppercase{tp}, the average number of covariates that are true positive; \MakeUppercase{fp}, the average number of covariates that are false positive.
\label{tabs2}
\end{table}

\section{Numerical experiments for covariate balancing difference-in-differences with optimal weighting matrix}\label{secB}

Table \ref{tabs3} lists results for data generated as in Case 2-1, Case 2-2, and Case 2-3 in Section \ref{sec4_4} and comparing the proposed criterion and the \MakeUppercase{qic}$_{\text{\MakeUppercase{w}}}$ through the forward selection method when propensity scores are estimated using the \MakeUppercase{cbd} method with the optimal matrix as the weighting matrix in the generalized method of moments.
The performance of the proposed model selection is similar to that when using the \MakeUppercase{cbd} method with the identity matrix.

\begin{table}
\renewcommand{\baselinestretch}{1.5}\selectfont
\caption{Comparison of the proposed criterion and the \MakeUppercase{qic}$_{\text{\MakeUppercase{w}}}$ through model selection for case in which the propensity scores are estimated using the \MakeUppercase{cbd} method with the optimal matrix as the weighting matrix.}
\begin{center}
\begin{tabular}{rrrrrrrrrrrr}
 & & & \multicolumn{3}{c}{Case 2-1} & \multicolumn{3}{c}{Case 2-2} & \multicolumn{3}{c}{Case 2-3} \\
 \multicolumn{1}{c}{$\beta^*$} & $n$ & & \multicolumn{1}{c}{Risk} & \multicolumn{1}{c}{\MakeUppercase{tp}} & \multicolumn{1}{c}{\MakeUppercase{fp}} & \multicolumn{1}{c}{Risk} & \multicolumn{1}{c}{\MakeUppercase{tp}} & \multicolumn{1}{c}{\MakeUppercase{fp}} & \multicolumn{1}{c}{Risk} & \multicolumn{1}{c}{\MakeUppercase{tp}} & \multicolumn{1}{c}{\MakeUppercase{fp}} \\
\addlinespace[1ex]
 & \multirow{2}{*}{200} & \multicolumn{1}{c}{Proposal} & 7.36 & 0.23 & 0.51 & 13.99 & 0.44 & 0.40 & 17.80 & 0.41 & 0.76 \\
 & & \multicolumn{1}{c}{\MakeUppercase{qic}$_{\text{\MakeUppercase{w}}}$} & 10.34 & 0.50 & 1.52 & 17.86 & 0.92 & 0.90 & 24.66 & 0.89 & 1.84 \\
\addlinespace[1ex]
 \multirow{2}{*}{0.1} & \multirow{2}{*}{400} & \multicolumn{1}{c}{Proposal} & 7.05 & 0.21 & 0.49 & 12.80 & 0.47 & 0.34 & 17.17 & 0.49 & 0.68 \\
 & & \multicolumn{1}{c}{\MakeUppercase{qic}$_{\text{\MakeUppercase{w}}}$} & 9.75 & 0.51 & 1.43 & 16.27 & 0.96 & 0.91 & 22.47 & 0.97 & 1.77 \\
\addlinespace[1ex]
 & \multirow{2}{*}{600} & \multicolumn{1}{c}{Proposal} & 7.45 & 0.23 & 0.49 & 13.65 & 0.49 & 0.37 & 16.54 & 0.50 & 0.66 \\
 & & \multicolumn{1}{c}{\MakeUppercase{qic}$_{\text{\MakeUppercase{w}}}$} & 10.30 & 0.52 & 1.49 & 16.79 & 1.01 & 0.93 & 22.08 & 1.00 & 1.81 \\
\addlinespace[1ex]
 & \multirow{2}{*}{200} & \multicolumn{1}{c}{Proposal} & 10.48 & 0.53 & 0.52 & 21.70 & 1.11 & 0.39 & 27.12 & 1.12 & 0.74 \\
 & & \multicolumn{1}{c}{\MakeUppercase{qic}$_{\text{\MakeUppercase{w}}}$} & 13.18 & 0.76 & 1.63 & 22.60 & 1.53 & 1.00 & 31.88 & 1.52 & 2.04 \\
\addlinespace[1ex]
 \multirow{2}{*}{0.5} & \multirow{2}{*}{400} & \multicolumn{1}{c}{Proposal} & 10.14 & 0.73 & 0.47 & 20.38 & 1.55 & 0.32 & 26.47 & 1.51 & 0.67 \\
 & & \multicolumn{1}{c}{\MakeUppercase{qic}$_{\text{\MakeUppercase{w}}}$} & 12.34 & 0.89 & 1.59 & 20.69 & 1.81 & 0.99 & 29.65 & 1.79 & 1.95 \\
\addlinespace[1ex]
 & \multirow{2}{*}{600} & \multicolumn{1}{c}{Proposal} & 10.57 & 0.84 & 0.51 & 20.45 & 1.76 & 0.38 & 24.65 & 1.76 & 0.66 \\
 & & \multicolumn{1}{c}{\MakeUppercase{qic}$_{\text{\MakeUppercase{w}}}$} & 12.93 & 0.95 & 1.60 & 21.27 & 1.92 & 1.01 & 28.69 & 1.92 & 2.05 \\
\addlinespace[1ex]
 & \multirow{2}{*}{200} & \multicolumn{1}{c}{Proposal} & 13.98 & 0.85 & 0.53 & 26.72 & 1.78 & 0.37 & 35.39 & 1.79 & 0.73 \\
 & & \multicolumn{1}{c}{\MakeUppercase{qic}$_{\text{\MakeUppercase{w}}}$} & 17.78 & 0.94 & 1.77 & 28.94 & 1.92 & 1.10 & 42.80 & 1.94 & 2.18 \\
\addlinespace[1ex]
 \multirow{2}{*}{1.0} & \multirow{2}{*}{400} & \multicolumn{1}{c}{Proposal} & 11.73 & 0.98 & 0.49 & 22.13 & 1.98 & 0.33 & 30.84 & 1.98 & 0.65 \\
 & & \multicolumn{1}{c}{\MakeUppercase{qic}$_{\text{\MakeUppercase{w}}}$} & 16.01 & 1.00 & 1.74 & 26.60 & 2.00 & 1.11 & 40.41 & 1.99 & 2.15 \\
\addlinespace[1ex]
 & \multirow{2}{*}{600} & \multicolumn{1}{c}{Proposal} & 12.21 & 1.00 & 0.49 & 22.92 & 2.00 & 0.37 & 28.63 & 2.00 & 0.63 \\
 & & \multicolumn{1}{c}{\MakeUppercase{qic}$_{\text{\MakeUppercase{w}}}$} & 16.81 & 1.00 & 1.72 & 27.77 & 2.00 & 1.11 & 39.30 & 2.00 & 2.19 \\
\addlinespace[1ex]
 & \multirow{2}{*}{200} & \multicolumn{1}{c}{Proposal} & 35.73 & 0.99 & 0.52 & 63.10 & 2.00 & 0.34 & 92.37 & 2.00 & 0.64 \\
 & & \multicolumn{1}{c}{\MakeUppercase{qic}$_{\text{\MakeUppercase{w}}}$} & 49.67 & 1.00 & 1.96 & 83.16 & 2.00 & 1.21 & 134.42 & 2.00 & 2.37 \\
\addlinespace[1ex]
 \multirow{2}{*}{3.0} & \multirow{2}{*}{400} & \multicolumn{1}{c}{Proposal} & 29.78 & 1.00 & 0.49 & 59.41 & 2.00 & 0.32 & 91.98 & 2.00 & 2.39 \\
 & & \multicolumn{1}{c}{\MakeUppercase{qic}$_{\text{\MakeUppercase{w}}}$} & 43.98 & 1.00 & 1.93 & 79.86 & 2.00 & 1.16 & 134.37 & 2.00 & 2.39 \\
\addlinespace[1ex]
 & \multirow{2}{*}{600} & \multicolumn{1}{c}{Proposal} & 30.75 & 1.00 & 0.47 & 61.53 & 2.00 & 0.34 & 88.53 & 2.00 & 0.65 \\
 & & \multicolumn{1}{c}{\MakeUppercase{qic}$_{\text{\MakeUppercase{w}}}$} & 45.42 & 1.00 & 1.95 & 83.30 & 2.00 & 1.20 & 129.84 & 2.00 & 2.38 \\
\end{tabular}
\end{center}
Risk, the empirical estimate of the risk of the selected model; \MakeUppercase{tp}, the average number of covariates that are true positive; \MakeUppercase{fp}, the average number of covariates that are false positive.
\label{tabs3}
\end{table}

\section*{Acknowledgement}
Yoshiyuki Ninomiya gratefully acknowledges support from JSPS KAKENHI (grant numbers 23H00809 and 23K18471). 
The authors declare no conflicts of interest associated with this manuscript. 
The LaLonde dataset is available in the R package Matching version 4.10.14 (\citealt{sek11multi}). 
Example R scripts used for the simulations and real data analyses are available from the authors.

\bibliography{List}

\end{document}